\documentclass[11pt]{article}

\usepackage[utf8]{inputenc}
\usepackage[margin=1in]{geometry}
\usepackage{amsmath}
\usepackage{amssymb}
\usepackage{amsthm}
\usepackage{mathtools}
\usepackage{bbm}
\usepackage{bm}
\usepackage{color}
\usepackage{graphicx}
\usepackage{array}
\usepackage{multirow}
\usepackage{enumitem}
\usepackage[ruled,vlined]{algorithm2e}
\usepackage{relsize}
\usepackage[dvipsnames]{xcolor}
\usepackage[linktocpage=true,
	pagebackref=true,colorlinks,
	linkcolor=BrickRed,citecolor=blue]
	{hyperref}
\usepackage{cleveref}
\usepackage[nottoc]{tocbibind}
\usepackage[normalem]{ulem}

\usepackage{thmtools,thm-restate}
\usepackage{nicefrac}
\usepackage{makecell}

\newcommand{\vb}{\boldsymbol}

\newcommand{\amplify}{\mathord{\raisebox{0.0ex}{\scalebox{0.75}{$+$}}}}

\DeclareMathOperator*{\One}{\mathbbm{1}}
\DeclareMathOperator*{\E}{\mathbb{E}}

\DeclareMathOperator{\diff}{\triangle}
\DeclareMathOperator{\contract}{/}
\DeclareMathOperator{\given}{\vert}

\DeclareMathOperator{\bbgiven}{\bigg\vert}
\DeclareMathOperator{\mutual}{;\,}

\DeclareMathOperator*{\TV}{TV}
\DeclareMathOperator*{\Ent}{\mathbb{H}}

\DeclareMathOperator*{\I}{\mathbb{I}}

\DeclareMathOperator*{\argmax}{argmax}

\DeclareMathOperator{\ALG}{ALG}

\DeclareMathOperator{\MECH}{MECH}

\DeclareMathOperator{\poly}{poly}
\DeclareMathOperator{\Law}{Law}

\DeclareMathOperator{\Unif}{Unif}
\DeclareMathOperator{\sign}{sign}

\newcommand{\DSICgs}{\ensuremath{\mathsf{DSIC_{GS}}}}
\newcommand{\DSICxos}{\ensuremath{\mathsf{DSIC_{XOS}}}}
\newcommand{\DSICgen}{\ensuremath{\mathsf{DSIC_{GEN}}}}
\newcommand{\GS}{\ensuremath{\mathsf{GS}}}
\newcommand{\XOS}{\ensuremath{\mathsf{XOS}}}
\newcommand{\GEN}{\ensuremath{\mathsf{GEN}}}
\newcommand{\VAL}{\ensuremath{\mathsf{VAL}}}

\newcommand{\calA}{\mathcal{A}}
\newcommand{\calB}{\mathcal{B}}

\newcommand{\calD}{\mathcal{D}}
\newcommand{\calE}{\mathcal{E}}
\newcommand{\calF}{\mathcal{F}}

\newcommand{\calH}{\mathcal{H}}
\newcommand{\calI}{\mathcal{I}}

\newcommand{\calM}{\mathcal{M}}
\newcommand{\calN}{\mathcal{N}}

\newcommand{\calP}{\mathcal{P}}

\newcommand{\calR}{\mathcal{R}}

\newcommand{\calU}{\mathcal{U}}
\newcommand{\calV}{\mathcal{V}}

\newcommand{\calX}{\mathcal{X}}
\newcommand{\calY}{\mathcal{Y}}

\let\abs\relax
\DeclarePairedDelimiter{\abs}{\lvert}{\rvert}

\newcommand{\IGNORE}[1]{}

\newcounter{note}[section]

\newcommand{\PreserveBackslash}[1]{\let\temp=\\#1\let\\=\temp}
\newcolumntype{C}[1]{>{\PreserveBackslash\centering}p{#1}}
\newcolumntype{R}[1]{>{\PreserveBackslash\raggedleft}p{#1}}
\newcolumntype{L}[1]{>{\PreserveBackslash\raggedright}p{#1}}

\newcommand{\ols}[1]{\mskip.5\thinmuskip\overline{\mskip-.5\thinmuskip {#1} \mskip-.5\thinmuskip}\mskip.5\thinmuskip} 
\newcommand{\olsi}[1]{\,\overline{\!{#1}}} 
\makeatletter
\newcommand\closure[1]{
  \tctestifnum{\count@stringtoks{#1}>1} 
  {\ols{#1}} 
  {\olsi{#1}} 
}
\long\def\count@stringtoks#1{\tc@earg\count@toks{\string#1}}
\long\def\count@toks#1{\the\numexpr-1\count@@toks#1.\tc@endcnt}
\long\def\count@@toks#1#2\tc@endcnt{+1\tc@ifempty{#2}{\relax}{\count@@toks#2\tc@endcnt}}
\def\tc@ifempty#1{\tc@testxifx{\expandafter\relax\detokenize{#1}\relax}}
\long\def\tc@earg#1#2{\expandafter#1\expandafter{#2}}
\long\def\tctestifnum#1{\tctestifcon{\ifnum#1\relax}}
\long\def\tctestifcon#1{#1\expandafter\tc@exfirst\else\expandafter\tc@exsecond\fi}
\long\def\tc@testxifx{\tc@earg\tctestifx}
\long\def\tctestifx#1{\tctestifcon{\ifx#1}}
\long\def\tc@exfirst#1#2{#1}
\long\def\tc@exsecond#1#2{#2}
\makeatother

\newtheorem{theorem}{Theorem}[section]

\Crefname{claim}{Claim}{Claims}
\newtheorem{proposition}[theorem]{Proposition}
\newtheorem{lemma}[theorem]{Lemma}
\newtheorem{corollary}[theorem]{Corollary}

\newtheorem{observation}[theorem]{Observation}
\newtheorem{fact}[theorem]{Fact}
\Crefname{fact}{Fact}{Facts}
\theoremstyle{definition}

\newtheorem{definition}[theorem]{Definition}

\makeatletter

\renewcommand*{\theHtheorem}{\theHsection.\arabic{theorem}}

\@for\thmenv:=%
  claim,%
  proposition,%
  lemma,%
  corollary,%
  conjecture,%
  observation,%
  fact,%
  example,%
  question,%
  definition,%
  assumption,%
  remark,%
  problem,%
  intuition%
\do{%
  \expandafter\let
    \csname theH\thmenv\endcsname
    \theHtheorem
}

\makeatother

\title{The Hardness of Dominant Strategy Mechanism Design, Revisited}
\author{Frederick V.\ Qiu}

\allowdisplaybreaks[3]

\begin{document}

\maketitle

\begin{abstract}
    We study the communication complexity of \emph{dominant-strategy incentive-compatible} (DSIC) mechanisms for combinatorial auctions. For $\gamma \in [\log m, m]$, let $\DSICgen(m, \gamma)$, $\DSICxos(m, \gamma)$, and $\DSICgs(m, \gamma)$ denote the best approximation ratio attainable by a deterministic, individually rational, no-negative-transfers, DSIC mechanism using at most $2^\gamma$ communication over $m$ items, for general monotone, XOS, and gross substitutes (GS) valuations, respectively. We give a unified proof that shows $\DSICgen(m, \gamma) = \Omega(m/\gamma)$, $\DSICxos(m, \gamma) = \Omega((m/\gamma)^{1/5})$, and $\DSICgs(m, \gamma) = \Omega((m/\gamma)^{1/7})$. The GS lower bound answers an open question of~\cite{DobzinskiRV22}: although polynomial-communication welfare maximization for GS valuations admits a polynomial-communication deterministic truthful mechanism via VCG, no good approximation is possible in polynomial-communication for deterministic DSIC mechanisms. Additionally, the general lower bound establishes that $\DSICgen(m, \gamma) = \Theta(m/\gamma)$ due to the deterministic DSIC $O(m/\gamma)$-approximation of~\cite{QiuW24}.

    We also obtain several results in the two-bidder setting. We show that attaining a $1.0001$-approximation for two weighted matroid-rank valuations (a subclass of GS) with a universally DSIC mechanism requires exponential communication. On the other hand, we give a polynomial-communication $(1+\sqrt{5})/2 \approx 1.618$-approximation for two submodular bidders using a universally DSIC mechanism. Prior to this work, it was not known whether even a polynomial-communication universally truthful mechanism could beat a $2$-approximation for two submodular valuations, nor whether a polynomial-communication universally DSIC mechanism could beat a $2$-approximation for two weighted matroid rank valuations.
\end{abstract}

\section{Introduction}

In a combinatorial auction, there is a set of items $M \coloneqq [m]$ we wish to allocate among $n$ bidders. Each bidder $i \in [n]$ holds a private monotone valuation function $v_i : 2^M \to \mathbb{R}_+$. An \emph{allocation} is an $n$-tuple of pairwise disjoint subsets of $M$. The objective is to find an allocation $(S_1, \dots, S_n)$ that approximately maximizes the social welfare $\sum_{i \in [n]} v_i(S_i)$.

Since each bidder $i$ only knows their own valuation $v_i$, two challenges arise:
\begin{enumerate}[topsep=1pt,itemsep=1pt]
    \item The bidders must communicate in order to find a high welfare allocation.
    
    \item A bidder may attempt to strategically manipulate the protocol by lying about their valuation, reducing overall welfare in order to achieve higher personal utility.
\end{enumerate}

Combinatorial auctions are well understood with respect to challenge (1). Ignoring incentives, polynomial-communication algorithms can attain a nearly tight $\Theta(\sqrt m)$-approximation for general valuations~\cite{LehmannOS02,NisanS06,BlumrosenN05a}, tight and near-tight $\Theta(1)$-approximations for subadditive, XOS, and submodular valuations~\cite{Feige09,DobzinskiNS10,EzraFNTW19,DobzinskiV13,FeigeV10}, and exact welfare maximization for gross substitutes (GS) valuations~\cite{NisanS06}.

Combinatorial auctions are also well understood with respect to challenge (2). A \emph{mechanism} outputs both an allocation $(S_1, \dots, S_n)$ and prices $(p_1, \dots, p_n)$; for bidders seeking to maximize their utility $v_i(S_i) - p_i$, a carefully designed price can incentivize them to behave \emph{truthfully} in the protocol. The VCG mechanism~\cite{Vickrey61,Clarke71,Groves73} shows how to design a truthful mechanism given a welfare-maximizing algorithm. However, welfare-maximization (and hence the VCG auction) requires exponential communication for all previously mentioned domains except GS.

When both challenges are considered together, surprisingly little is known. The best known deterministic truthful mechanisms achieve an $O(m/\log m)$-approximation for general valuations and $O(\sqrt{m/\log m})$-approximation for subadditive, XOS, and submodular valuations~\cite{QiuW24}, nearly a factor of $\sqrt{m}$ worse than their fully algorithmic counterparts. Meanwhile, lower bounds separating deterministic truthful mechanisms from algorithms are known only in restricted settings~\cite{AssadiKSW20,RonTWZ24}. As such, understanding the interplay between efficiency and incentive compatibility is a central problem at the intersection of Economics and Computing.

\subsection{Notions of Truthfulness}

To be exact, a mechanism describes a communication protocol and prescribes a strategy for each bidder depending on their valuation. In the combinatorial auctions literature, we say that a mechanism is \emph{truthful} if each bidder following their prescribed strategy results in an ex-post Nash equilibrium. In that sense, truthfulness is the guarantee that each bidder is best off playing truthfully, assuming the other bidders are playing consistently with some prescribed strategy.

However, we could also consider a \emph{dominant-strategy incentive-compatible} (DSIC) mechanism, where each bidder is best off playing truthfully, \emph{no matter how} the other bidders behave, be it rationally, adversarially, nonsensically, etc.

A natural question is whether DSIC mechanism design is inherently more challenging than truthful mechanism design.~\cite{RubinsteinSTWZ21} study precisely this question, and are the first to show that there exist (artificially constructed) domains where DSIC mechanisms require exponentially more communication than truthful mechanisms.~\cite{DobzinskiRV22} later studied this question in the natural setting of combinatorial auctions. Among other results, they prove that
\begin{itemize}[topsep=1pt,itemsep=1pt]
    \item No polynomial-communication DSIC mechanism can maximize the welfare for GS valuations.
    \item No polynomial-communication deterministic DSIC mechanism can attain an $m^{1-o(1)}$-approx-imation for general valuations.
\end{itemize}

Welfare maximization in polynomial communication is possible for GS valuations, so VCG yields truthful mechanisms doing the same. Thus, the first result directly extends the result of~\cite{RubinsteinSTWZ21} to combinatorial auction domains.

The second result will be the focus of this paper: it gives a hardness-of-approximation result for deterministic DSIC mechanisms for general valuations.~\cite{DobzinskiRV22} leaves open whether it is also hard to approximate the welfare for GS valuations with a deterministic DSIC mechanism.

Note that while our main results will concern \emph{deterministic} DSIC mechanisms, we will also study \emph{universally} DSIC, a randomized analogue of DSIC that we will formally define in~\Cref{sec:Preliminaries}. Randomization allows us to design more powerful mechanisms, but there are reasons to study both the randomized and deterministic settings. We defer further discussion to the related work section.

\subsection{Results}

For a parameter $\gamma$, let $\DSICgs(m, \gamma)$, $\DSICxos(m, \gamma)$, and $\DSICgen(m, \gamma)$ denote the optimal approximation ratio attainable on $m$ items with $2^\gamma$ communication with a deterministic, individually rational (IR), no-negative-transfers (NNT), DSIC mechanism for GS valuations, XOS valuations, and general valuations, respectively.\footnote{Since our GS construction only uses convolutions of weighted matroid rank valuations, the lower bound also holds for the class of Rado valuations, which are a strict subclass of GS.}

\begin{theorem} \label{thm:GSDSICHardness}
    $\DSICgs(m, \gamma) = \Omega((m/\gamma)^{1/7})$ for $\gamma \in [\log m, m]$.
\end{theorem}

\begin{theorem} \label{thm:XOSDSICHardness}
    $\DSICxos(m, \gamma) = \Omega((m/\gamma)^{1/5})$ for $\gamma \in [\log m, m]$.
\end{theorem}

\begin{theorem} \label{thm:GeneralDSICHardness}
    $\DSICgen(m, \gamma) = \Theta(m/\gamma)$ for $\gamma \in [\log m, m]$.
\end{theorem}

Hence, we obtain the first super-constant hardness-of-approximation results for deterministic DSIC mechanisms on canonical complement-free valuation classes. In particular, the GS result answers the open question of~\cite{DobzinskiRV22}. Additionally, the hardness-of-approximation for general valuations is tight up to constant factors, due to the deterministic DSIC $O(m/\gamma)$-approximation by~\cite{QiuW24}.

We also point out that the $O(m/\gamma)$-approximation is DSIC, maximal-in-range (MIR), value query implementable, and computationally efficient~\cite{QiuW24}. Despite the fact that we have no non-trivial lower bounds for the full class of deterministic truthful mechanisms for general valuations, we now know that even requiring any one of these four conditions (DSIC, MIR, value query implementable, computationally efficient) immediately runs into a tight impossibility result~\cite{QiuW24,BlumrosenN05a,DanielySS15}.

We remark that a single unified proof yields all three of the above theorems.

The hard instances used in our main results require many bidders and do not generalize to few-bidder settings. We therefore provide hardness-of-approximation results in the two-bidder setting.

\begin{theorem} \label{thm:TwoBidderSimultaneousHardness}
    For every constant $\varepsilon \in (0, 1/16]$, any simultaneous algorithm attaining a $(15/16+\varepsilon)$-approximation for two matroid rank valuations uses exponential communication.
\end{theorem}

\begin{theorem} \label{thm:TwoBidderDSICHardness}
    Any IR, NNT, universally DSIC mechanism attaining a $0.9999$-approximation for two weighted matroid rank valuations uses exponential communication.
\end{theorem}

Note that matroid rank and weighted matroid rank valuations are GS, and hence admit deterministic welfare-maximizing algorithms/truthful mechanisms.

All of our DSIC hardness-of-approximation results are reductions from hard instances for simultaneous algorithms, so~\Cref{thm:TwoBidderSimultaneousHardness} is a standalone result which also provides intuition for~\Cref{thm:TwoBidderDSICHardness}. We also remark that unlike our many-bidder lower bounds,~\Cref{thm:TwoBidderDSICHardness} holds for universally DSIC mechanisms. Therefore, even though good universally DSIC mechanisms exist for many bidders, such as the polynomial-communication $(\log\log m)^2$-approximation of~\cite{AssadiKS21} for XOS valuations, it is still impossible for a polynomial-communication universally DSIC mechanism to perform arbitrarily well on two bidders.

Finally, we prove two positive results in the two-bidder setting.

\begin{theorem} \label{thm:TwoBidderSimultaneousAlgorithm}
    There exists a $(41-\sqrt{6})/50 \approx 0.771$-approximation for two matroid rank valuations with a simultaneous algorithm using polynomial communication.
\end{theorem}

\begin{theorem} \label{thm:TwoBidderDSICMechanism}
    There exists a $(\sqrt{5}-1)/2 \approx 0.618$-approximation for two submodular valuations with an IR, NNT, universally DSIC mechanism using polynomial communication.
\end{theorem}

\Cref{thm:TwoBidderSimultaneousAlgorithm} shows that it is possible to beat the $3/4$-approximation of~\cite{BravermanMW18} for binary XOS valuations on the class of binary submodular valuations (which is exactly the class of matroid rank valuations).~\Cref{thm:TwoBidderDSICMechanism} gives the first polynomial-communication DSIC mechanism to beat a $1/2$-approximation for two weighted matroid rank valuations, and also the first polynomial-communication truthful mechanism to beat a $1/2$-approximation for two submodular valuations.

Our many-bidder results are presented in the body of the paper, while the proofs of the two-bidder results can be found in the appendix.

\subsection{Technical Highlights}

\paragraph{Simultaneous Algorithm for Two Matroid Rank Valuations.}
The simultaneous algorithm to beat $3/4$ is a combination of two existing ideas.~\cite{FeigeV10} gives an interactive $5/6$-approximation for two submodular valuations (under certain conditions) via a clever fractional solution rounding scheme. However, while the rounding scheme itself requires no further interaction after computing the fractional solution, it is not clear how to find a good fractional solution in the first place when restricted to simultaneous communication. On the other hand,~\cite{BravermanMW18} gives a simultaneous $3/4$-approximation for two XOS valuations via standard contention resolution, but this approach is capped at $3/4$ because we must lose a factor of $1/4$ to contention in the worst case (each bidder requests an item w.p.\ $1/2$).

The key observation is that the worst case $3/4$-approximation for standard contention resolution only occurs when the two matroids are close to uniformly dense. But when the matroids are uniformly dense, the fractional solution induced by each bidder sending a pair of as-disjoint-as-possible bases is not far from a true fractional solution, and so we can obtain nearly a $5/6$-approximation from the submodular rounding scheme. Mixing these two procedures with the appropriate probability allows us to beat a $3/4$-approximation on all profiles.

\vspace*{-1em}
\paragraph{DSIC Mechanism for Two Submodular Valuations.}
The idea here is again to mix two existing ideas that perform well on different kinds of profiles. One mechanism is just to randomly allocate the items, which obtains a $1/2$-approximation in the worst case. Another mechanism is the \emph{extended parallel second-price auction} (EP2A) from~\cite{FikiorisF20}, which asks for the marginal value of each item from both bidders. The higher bidder gains the ``right'' to purchase the item at the lower bidder's value, and both bidders purchase their favorite available set at the given prices.~\cite{FikiorisF20} prove that the EP2A is DSIC for subadditive valuations.

Observe that when the instance is close to modular, the EP2A sets nearly ``correct'' prices on the items, which allows the items to go to the correct bidder. The bad case is when the instance is far from modular; the reported marginal values may be much higher than how much actual value the bidders can extract, leaving many items unallocated. However, this is precisely the case where we expect to beat a $1/2$-approximation from a random allocation. Our contribution is showing that randomly allocating a portion of the items and running the EP2A on the remaining items balances between these two extremes and allows us to beat $1/2$.

Importantly, submodularity is necessary because by running the EP2A after an initial random allocation, we need the EP2A to be DSIC after contracting the original valuations. Submodularity is preserved under contraction, while subadditivity is not.

\vspace*{-1em}
\paragraph{DSIC Hardness-of-Approximation: Existing Framework.}
To put our techniques in context, we first describe the framework of~\cite{DobzinskiRV22} for proving DSIC hardness-of-approximation.

Suppose that Alice moves first in a DSIC mechanism and can send two different messages $X$ and $Y$. After observing Alice's message, Bob may use a \emph{contingent} strategy: for example, he may choose a continuation that is favorable to Alice after $X$ and unfavorable after $Y$. If this makes the strategy producing $X$ more attractive to a valuation prescribed to send $Y$, then DSIC is violated.

This obstruction appears in~\cite{RubinsteinSTWZ21} and is formalized for combinatorial auctions by~\cite{DobzinskiRV22}. At a high level, the mechanism can defend against such a contingent strategy in two ways. Alice's prescribed strategy may \emph{guarantee} her a good outcome even under Bob's adverse continuation, or Bob may be forced to \emph{commit} to his behavior before learning whether Alice sent $X$ or $Y$. Either way, the mechanism is limited in its ability to adapt later communication to earlier information.

Using this principle,~\cite{DobzinskiRV22} convert a good deterministic DSIC mechanism into a good simultaneous algorithm on a suitable hard distribution. A communication lower bound for the simultaneous problem then yields a lower bound for the DSIC mechanism.

\vspace*{-1em}
\paragraph{DSIC Hardness-of-Approximation: Existing Proof Sketch.}
The simultaneous hard instance divides the bidders into groups and gives every bidder an exponentially long list of sets. A bidder has value $1$ if they receive a set on their list and value $0$ otherwise. Each group contains one hidden special bidder who values a private special set. The special sets are mutually compatible, while every pair of non-special sets intersects. Thus, most of the optimal welfare comes from identifying the special bidder in each group.

\cite{DobzinskiRV22} duplicate these valuations across exponentially many multiplicative scales. After locating one bidder whose message distinguishes two nearby scales, they use a contingent strategy argument to show that every other bidder must send different first messages for some pair of scales in a range of length $O(\log m)$. Pigeonholing yields one common scale $s$ at which a $1/O(\log m)$ fraction of the relevant bidders send different first messages when they have $s v_i$ versus $2s v_i$.

Restrict attention to these bidders and the low scale. Since the DSIC mechanism gets a good approximation, it must correctly allocate a large fraction of the special sets. We apply a small random perturbation on the valuation of each bidder, so that when a bidder has positive value, they also have positive utility w.h.p.

Now, consider a bidder with positive utility. The DSIC inequalities and rigid structure of the 0-1 valuations show that when the valuation $s v_i$ plays the prescribed strategy of the valuation $2s v_i$, their utility is also positive. Therefore, consider the contingent strategy where the other bidders play normally after the message of $2s v_i$, but adversarially after the message of $s v_i$. If this adversarial continuation could give $s v_i$ utility $0$, then $s v_i$ would have a profitable deviation to the prescribed strategy of $2s v_i$. Therefore, every bidder with positive utility can \emph{guarantee} themselves positive utility against any continuation. \textbf{The rigid structure of the 0-1 valuations then implies that any bidder with positive utility receives their full value.}

Thus, the simultaneous algorithm asks each bidder to send their first message in the DSIC mechanism, and checks whether each bidder has a message-consistent strategy that guarantees them positive utility. Since non-special sets intersect, almost every bidder passing this test must obtain value from a special set, so the protocol can allocate each group's special set to its certified bidder.

\vspace*{-1em}
\paragraph{DSIC Hardness-of-Approximation: Our Contributions.}
Our proof retains the outer framework of~\cite{DobzinskiRV22}: we convert a good DSIC mechanism into a good simultaneous algorithm on a hard instance where most of the welfare comes from allocating special sets to hidden special bidders. Our GS and general valuation instances are parameterized refinements of their construction, while our XOS instance adapts ideas from~\cite{Assadi17}.

Our first change is to separate the scale of the underlying valuation from the value of a \emph{special outcome}. For GS and XOS valuations, the special outcome is receiving a high-value item $\dagger$; for general valuations, it is receiving the entire set of items $M$. The value of the special outcome is large enough to affect incentives for each individual bidder, but small compared with the total welfare of the hard instance. After locating one bidder whose first message distinguishes two nearby types, a contingent strategy argument forces every other bidder at a fixed valuation scale to change their first message at some threshold (possibly depending on the bidder and their valuation) of the special-outcome value. This allows us to work with all bidders at the same underlying valuation scale, avoiding the common-scale averaging loss in~\cite{DobzinskiRV22}.

The remaining barrier is quantitative. In the 0-1 construction of~\cite{DobzinskiRV22}, guaranteeing \emph{any} positive utility already certifies that a bidder receives their full value. This is false for GS and XOS valuations: a bidder may receive positive utility while recovering only a small fraction of their full value. We instead need to force a utility guarantee comparable to that full value.

We do this through \emph{amplification}. Each bidder has exponentially many candidate clauses, and amplification doubles the coefficients of a sparse collection of clauses. A short first message is unlikely to rule out the possibility that the hidden special coordinate was doubled. Thus, the observed message is typically also consistent with a counterfactual valuation that assigns double value to the special set. IR bounds the payment under the undoubled outcome, so DSIC forces the doubled valuation’s truthful strategy to obtain high utility. A second contingent strategy argument then upgrades this normal continuation utility into a high utility guarantee. The protocol then uses these quantitative guarantees to perform a randomized allocation that allocates the special sets correctly with good probability.

For XOS and general valuations, amplification is implemented directly by changing clause coefficients. For GS valuations, we implement the analogous operation by convolving two weighted matroid rank valuations, as GS is closed under convolution.

\subsection{Related Work}

\paragraph{The Communication Cost of Incentives.}
The study of how incentive constraints affect communication requirements was initiated by~\cite{FadelS09,BabaioffBS13}. Subsequent work established exponential separations between computing a social choice function and implementing it with incentives~\cite{DobzinskiR21,RubinsteinSTWZ21}. In particular,~\cite{RubinsteinSTWZ21} give the first exponential separation between truthfulness and DSIC, already in an artificial two-player domain. The contingent-strategy obstruction underlying their construction is formalized for combinatorial auctions by~\cite{DobzinskiRV22}. They prove that exact welfare maximization for GS valuations by a DSIC mechanism requires exponential communication, despite truthful mechanisms only needing polynomial communication~\cite{NisanS06}. They also prove an $m^{1-o(1)}$-approximation lower bound for deterministic DSIC mechanisms with general valuations. Our impossibility results build directly on this framework, extending the hardness-of-approximation to GS and XOS valuations, and obtaining a tight bound for general valuations witnessed by the DSIC mechanism of~\cite{QiuW24}.

\vspace*{-1em}
\paragraph{Simultaneous Combinatorial Auctions.}
Simultaneous combinatorial auctions were previously studied in~\cite{DobzinskiNO14,Assadi17}. Our simultaneous hard distributions are parameterized refinements of the constructions in~\cite{DobzinskiRV22,Assadi17}. For two bidders,~\cite{BravermanMW18} give a polynomial-communication simultaneous $3/4$-approximation for binary XOS valuations. Our positive two-bidder simultaneous result combines the contention resolution viewpoint of~\cite{BravermanMW18} with the submodular rounding framework of~\cite{FeigeV10}.

\vspace*{-1em}
\paragraph{Deterministic Truthful Mechanisms.}
For the standard notion of truthfulness,~\cite{Dobzinski16b} develops the taxation-complexity framework relating the communication of a truthful mechanism to the complexity of the menus it presents. Using this framework,~\cite{AssadiKSW20} prove the first separation between deterministic truthful mechanisms and algorithms, for two XOS bidders. More recently,~\cite{RonTWZ24} obtain the first such separation beyond two bidders. These works establish constant-factor separations for deterministic truthful mechanisms on either two-bidder or augmented valuation domains. In contrast, our many-bidder results give super-constant lower bounds for the stronger DSIC requirement on canonical domains like GS and XOS valuations.

The best known deterministic truthful mechanisms are maximal-in-range (MIR).~\cite{QiuW24} characterize the tradeoff between communication and approximation of MIR mechanisms and, for general valuations, obtain a tight $\Theta(m/\gamma)$-approximation with $2^\gamma$ communication. As their mechanism also happens to be DSIC, our general valuation lower bound shows that the same tradeoff is also optimal among all deterministic DSIC mechanisms.

\cite{FikiorisF20} describe the EP2A, which they show to be DSIC for subadditive valuations, and welfare-maximizing under additional assumptions. We use the EP2A as a building block for a universally DSIC mechanism for two submodular valuations.

\vspace*{-1em}
\paragraph{Randomized Truthful Mechanisms.}
There is a long line of work designing universally truthful mechanisms for complement-free valuations~\cite{DobzinskiNS06,Dobzinski07,KrystaV12,Dobzinski16a,AssadiS19,AssadiKS21}. Our two-bidder universally DSIC mechanism uses a modified version of the EP2A and a novel analysis to obtain an approximation guarantee for submodular valuations, without additional assumptions.

There are also computational and query-complexity limitations for randomized truthful mechanisms~\cite{DughmiV11,DobzinskiV16}, while~\cite{AssadiKSW20} gives a communication lower bound applying to universally truthful mechanisms for two XOS bidders. Our two-bidder result instead proves a communication lower bound for universally DSIC mechanisms on the substantially narrower class of weighted matroid rank valuations (which are a subclass of GS valuations).

\section{Preliminaries} \label{sec:Preliminaries}

We assume that all parameters are chosen such that all relevant quantities are integers (e.g., if we state that a set $S \subseteq M$ has size $m/10$, we implicitly assume $m$ is divisible by $10$).

In a combinatorial auction, there is a set of items $M \coloneqq [m]$ that we wish to allocate among $n$ bidders. Each bidder $i \in [n]$ holds a private valuation function $v_i : 2^M \to \mathbb{R}_+$. An \emph{allocation} is an $n$-tuple of pairwise disjoint subsets of $M$. The objective is to find an allocation $(S_1, \dots, S_n)$ that approximately maximizes the welfare $\sum_{i \in [n]} v_i(S_i)$. We often restrict the valuations $v_1, \dots, v_n$ to a valuation class $\calV$.

\begin{definition}[Allocation Rule]
    An \emph{allocation rule} is a (possibly random) mapping from profiles $(v_1, \dots, v_n)$ to allocations $(S_1, \dots, S_n)$.
\end{definition}

\begin{definition}[Approximation Ratio]
    An allocation rule attains an \emph{$\eta$-approximation} for $\calV$ if for all $v_1, \dots, v_n \in \calV$, the (possibly random) allocation $S_1, \dots, S_n$ output by the allocation rule satisfies
    \[
        \E\bigg[\sum_{i \in [n]} v_i(S_i)\bigg] \quad \geq \quad \frac{1}{\eta} \max_{(T_1, \dots, T_n)} \sum_{i \in [n]} v_i(T_i) \enspace .
    \]
    When $\eta \geq 1$ is a constant, we may state the approximation ratio instead as $1/\eta \leq 1$.
\end{definition}

\subsection{Communication Model}

Since each bidder holds a private valuation function, we can think of each bidder $i \in [n]$ as a participant in a communication protocol with input $v_i : 2^M \to \mathbb{R}_+$. As a technical detail, we assume that the valuations are restricted to a range $\calR \subseteq \mathbb{R}_+$ of size $2^{\poly(m)}$ so that every value $v_i(S)$ can always be precisely communicated. Without this assumption, we have trivial impossibilities stemming from issues of scale. We will use
\[
    \calR \quad = \quad \{2^{s/2} t/m : s \in [-2^m, 2^m] \cap \mathbb{Z},\, t \in [0, m^5] \cap \mathbb{Z}\} \enspace ,
\]
which supports values on a doubly-exponential scale while still allowing for $O(m)$-bit representations via the pair $(s, t)$.

We use the blackboard model of communication, i.e., all communication is visible to all parties. In particular, we can think of any communication protocol as a rooted tree. At each node of the tree, each bidder $i$ simultaneously sends a (possibly random) message from a list of possible messages, which may depend on the bidder and the node. We then move to a child node depending on the messages sent by each bidder. We assume that every distinct $n$-tuple of messages induces a distinct child node. Each leaf node of the protocol designates a fixed allocation. Every communication protocol has a \emph{strategy prescription}.

\begin{definition}[Strategy]
    A \emph{strategy} for bidder $i$ is a function from the nodes of a communication protocol to the message sent by bidder $i$ at the node. Note that since a strategy defines the message sent by bidder $i$ at every node, and because every distinct message history induces a distinct node in the communication protocol, a strategy is implicitly allowed to depend on the message history. A \emph{strategy prescription} is a function from valuations to strategies.
\end{definition}

\begin{definition}[Algorithm]
    An \emph{algorithm} describes a communication protocol whose output is the output of some allocation rule, when every bidder follows their prescribed strategy. An algorithm attains an \emph{$\eta$-approximation} if its allocation rule attains an $\eta$-approximation. The communication used by the algorithm is the worst-case length of the transcript of the communication protocol.
\end{definition}

We will be interested in the following class of algorithms.

\begin{definition}[Simultaneous Algorithm]
    A \emph{simultaneous algorithm} is an algorithm which describes a communication protocol where every child of the root node is a leaf. In other words, a simultaneous algorithm consists of only a single round of communication: each bidder sends a message, and then a referee decides the allocation purely as a function of the messages.
\end{definition}

\subsection{Truthful Mechanisms}

Until now, we have discussed combinatorial auctions as a purely algorithmic problem. But it is also possible to imagine that each bidder, instead of dutifully following the strategy prescribed by the algorithm, seeks to selfishly maximize their own utility. Then bidders may behave unexpectedly and prevent an algorithm from realizing the intended allocation rule.

To remedy this, we allow the leaf nodes (i.e., outcomes) of the communication protocol to include not only an allocation, but also \emph{payments} that each bidder will make.

\begin{definition}[Payment Rule]
    A \emph{payment rule} is a (possibly random) mapping from profiles $(v_1, \dots, v_n)$ to payments $(p_1, \dots, p_n) \in \mathbb{R}^n$.
\end{definition}

\begin{definition}[Mechanism]
    A \emph{mechanism} describes a communication protocol whose output is the output of some allocation and payment rule. A mechanism attains an \emph{$\eta$-approximation} if its allocation rule attains an $\eta$-approximation. The communication used by the mechanism is the worst-case length of the transcript of the communication protocol.
\end{definition}

\begin{definition}[Individually Rational]
    A mechanism is \emph{individually rational} (IR) if following the prescribed strategy guarantees non-negative utility on all profiles consisting of prescribed strategies.
\end{definition}

\begin{definition}[No-Negative-Transfers]
    A mechanism is \emph{no-negative-transfers} (NNT) if all payments are non-negative.
\end{definition}

We then assume that bidders have \emph{quasilinear} utility, that is, they seek to maximize the quantity $v_i(S_i) - p_i$. This then gives rise to the following notions.

\begin{definition}[Incentive-Compatible]
    A mechanism is \emph{incentive-compatible} (truthful) if it describes a communication protocol such that the profile of prescribed strategies forms an ex-post Nash equilibrium.
\end{definition}

\begin{definition}[Dominant-Strategy Incentive-Compatible]
    A mechanism is \emph{dominant-strategy incentive-compatible} (DSIC) if it describes a communication protocol such that it is a dominant strategy for each bidder to follow their prescribed strategy. In particular, the prescribed strategy is utility-maximizing on all strategy profiles, not merely all profiles of prescribed strategies.
\end{definition}

Observe that we have been a little sloppy in what we mean by Nash equilibrium and dominant strategy. For a randomized mechanism, does a bidder have a profitable deviation only when they can secure higher expected utility, or does it suffice for there to exist an outcome with higher utility? In this paper, we use the latter definition.

\begin{definition}[Universally DSIC]
    A mechanism is \emph{universally DSIC} if it can be written as a distribution over deterministic DSIC mechanisms.
\end{definition}

\subsection{Valuation Classes}

A valuation function is any function from sets of items to non-negative real numbers. All valuations are always assumed to be monotone non-decreasing. Aside from the class of general monotone valuations, our results discuss valuation classes from the \emph{complement-free hierarchy}:
\[
    \text{Gross Substitutes} \quad \subsetneq \quad \text{Submodular} \quad \subsetneq \quad \text{XOS} \quad \subsetneq \quad \text{Subadditive} \enspace .
\]
As pointed out previously, there exist deterministic truthful mechanisms that maximize the welfare over GS valuations. We will show that GS valuations do not even admit good approximations by polynomial-communication DSIC mechanisms. For these results, we will use the following subclasses of GS:
\[
    \text{Matroid Rank} \quad \subsetneq \quad \text{Weighted Matroid Rank} \quad \subsetneq \quad \text{Gross Substitutes} \enspace .
\]
We define the relevant classes below. Given $\vb{x} \in \mathbb{R}_+^m$ and $S \subseteq M$, denote $\vb{x}(S) = \sum_{j \in S} x_j$.

\begin{definition}[Matroid]
    A set family $\calM \subseteq 2^M$\footnote{Convention typically dictates that a matroid is notated $\calM = (M, \calI)$ where $\calI \subseteq 2^M$, but we directly notate a matroid $\calM \subseteq 2^M$ for brevity.} is a \emph{matroid} if it satisfies
    \begin{itemize}[topsep=1pt,itemsep=1pt]
        \item $\calM \ne \emptyset$.
        \item Downward-Closure: if $T \in \calM$, then all $S \subseteq T$ satisfy $S \in \calM$ (downward-closure).
        \item Augmentation: for all $S, T \in \calM$ where $\abs{S} < \abs{T}$, there exists $j \in T \setminus S$ such that $S \cup \{j\} \in \calM$.
    \end{itemize}
\end{definition}

\begin{definition}[Matroid Rank Function]
    Given a matroid $\calM \subseteq 2^M$, its \emph{rank function} $r : 2^M \to \mathbb{Z}_+$ is defined $r(S) = \max_{T \in \calM} \abs{S \cap T}$. The \emph{rank} of $\calM$ is the value $r(M)$.
\end{definition}

\begin{definition}[Weighted Matroid Rank]
    A valuation $v : 2^M \to \mathbb{R}_+$ is \emph{weighted matroid rank} if there exists a matroid $\calM \subseteq 2^M$ and item weights $\vb{w} \in \mathbb{R}_+^m$ such that $v(S) = \max_{T \in \calM} \vb{w}(S \cap T)$.
\end{definition}

\begin{definition}[Convolution]
    For two functions $f, g : 2^M \to \mathbb{R}_+$, their \emph{convolution} $f \star g$ is defined $(f \star g)(S) = \max_{T \subseteq S} \{f(S \setminus T) + g(T)\}$.
\end{definition}

\begin{fact}
    Gross substitutes is preserved under convolution.
\end{fact}

\begin{definition}[Submodular]
    A valuation $v : 2^M \to \mathbb{R}_+$ is \emph{submodular} if for all $j \in M$, the marginal value $v(S \cup \{j\}) - v(S)$ is monotone non-increasing in $S$.
\end{definition}

\begin{definition}[XOS]
    A valuation $v : 2^M \to \mathbb{R}_+$ is \emph{XOS} if there exist item weights $\vb{w}_1, \dots, \vb{w}_\lambda \in \mathbb{R}_+^m$ such that $v(S) = \max_{\ell \in [\lambda]} \vb{w}_\ell(S)$.
\end{definition}

\subsection{Matroids}

We will use the following result of~\cite{BalcanH11}. Note that this is not directly stated in their Theorem 7, but is directly implied from the proof.

\begin{theorem}[\cite{BalcanH11}] \label{thm:BalcanHarvey}
    Let $b, d$ be parameters such that $b \in [20, d]$, and let $\beta \geq 96d^3/b^2$ and $\lambda \in [4d/b, 2^{b/10}]$. Then there exist sets $C_1, \dots, C_\lambda \subseteq [\beta]$ of size $d$ such that for all $\vb{x} \in \{0, 1\}^\lambda$, there exists a matroid rank function $r_{\vb{x}} : 2^{[\beta]} \to \mathbb{R}_+$ such that
    \begin{itemize}[topsep=1pt,itemsep=1pt]
        \item $r_{\vb{x}}([\beta]) \leq d$.
        \item $r_{\vb{x}}(C_\ell) = d$ if $x_\ell = 1$.
        \item $r_{\vb{x}}(C_\ell) = b$ if $x_\ell = 0$.
    \end{itemize}
\end{theorem}

\subsection{Information Theory}

Below are some information theoretic results we use to prove communication lower bounds for simultaneous algorithms. We use $\log$ to denote the binary logarithm.

\begin{definition}[Total Variation Distance]
    The \emph{total variation distance} of two distributions $\calD, \calD'$ supported on $\calX$ is defined
    \[
        \TV(\calD, \calD') \quad = \quad \max_{\calE \subseteq \calX} \bigg\vert\sum_{x \in \calE} \Big(\Pr_{X \sim \calD}[X = x] - \Pr_{X \sim \calD'}[X = x]\Big)\bigg\vert \enspace .
    \]
\end{definition}

\begin{definition}[Entropy]
    For discrete random variables $X, Y$ supported on $\calX, \calY$, the \emph{entropy} of $X$ conditioned on $Y$ is defined
    \[
        \Ent(X \given Y) \quad = \quad -\sum_{y \in \calY} \sum_{x \in \calX} \Pr[X = x, Y = y] \log(\Pr[X = x \given Y = y]) \quad \leq \quad \log\abs{\calX} \enspace .
    \]
\end{definition}

\begin{definition}[Mutual Information]
    For discrete random variables $X, Y, Z$, the \emph{mutual information} between $X$ and $Y$ conditioned on $Z$ is defined
    \[
        \I(X \mutual Y \given Z) \quad = \quad \Ent(X \given Z) - \Ent(X \given Y, Z) \quad = \quad \Ent(Y \given Z) - \Ent(Y \given X, Z) \enspace .
    \]
\end{definition}

\begin{fact}[Superadditivity of Mutual Information]
    Let $X_1, \dots, X_n, Y, Z$ be discrete random variables, where $X_1, \dots, X_n$ are independent conditioned on $Z$. Then
    \[
        \I(X_1, \dots, X_n \mutual Y \given Z) \quad \geq \quad \sum_{i \in [n]} \I(X_i \mutual Y \given Z) \enspace .
    \]
\end{fact}

\begin{definition}
    For discrete random variables $X, Y$, let the distribution of $X$ conditioned on $Y = y$ be denoted $\Law(X \given Y = y)$.
\end{definition}

\begin{fact}[Pinsker's Inequality]
    For discrete random variables $X, Y, Z$,
    \[
        \E[\TV(\Law(X \given Y, Z), \Law(X \given Z))^2] \quad \leq \quad \frac{\ln 2}{2} \I(X \mutual Y \given Z) \enspace .
    \]
\end{fact}

\section{Known-Set Hard Instances for Simultaneous Algorithms}

Following the approach of~\cite{DobzinskiRV22}, our DSIC lower bounds show that a good DSIC mechanism allows for the construction of a good simultaneous algorithm on a particular hard distribution of instances. In this section, we formally describe the hard distribution and prove a simultaneous communication lower bound. We now define the hard distribution abstractly and state the necessary properties. This is so that we can give a unified communication lower bound and DSIC reduction for the parameterized instance, then separately optimize the parameters for each valuation class.

\begin{definition}[Known-Set Simultaneous Hard Distribution] \label{def:KSSHD}
    A \emph{known-set simultaneous hard distribution} with parameters $\alpha, \eta, \lambda, \vb{C}, \calU$, abbreviated $(\alpha,\eta,\lambda,\vb{C},\calU)$-KSSHD, is a distribution over combinatorial auctions instances with the following properties:
    \begin{enumerate}[topsep=1pt,itemsep=1pt]
        \item Let $\beta = m/\eta$ and $\kappa = m/(2\alpha)$.
        \item Let $\vb{C} = (C_1, \dots, C_\lambda)$, where $C_\ell \subseteq [\beta]$, $\abs{C_\ell} = \alpha$ for all $\ell \in [\lambda]$.
        \item Let $\calU = \{u_{\vb{x}} : \vb{x} \in \{0,1,2\}^\lambda\}$,\footnote{The class $\calU$ is stated for $\vb{x} \in \{0,1,2\}^\lambda$ because of what we need in the DSIC hardness, but we only use $\vb{x} \in \{0,1\}^\lambda$ for the simultaneous hardness.} where for all $\vb{x} \in \{0,1,2\}^\lambda$, $u_{\vb{x}} : 2^{[\beta]} \to \mathbb{Z}_+$ is either
        \begin{enumerate}[topsep=0pt,itemsep=1pt]
            \item A GS or XOS valuation satisfying:
            \begin{enumerate}[topsep=0pt,itemsep=1pt]
                \item $u_{\vb{x}}(S) \leq \min\{3\alpha, 2\abs{S}\}$ for all $S \subseteq [\beta]$.
                \item $u_{\vb{x}}(C_\ell) \geq x_\ell\alpha$.
                \item $u_{\vb{x}}(C_\ell) \leq 30\alpha/\eta^2$ if $x_\ell = 0$.
            \end{enumerate}
            \item A general valuation satisfying:
            \begin{enumerate}[topsep=0pt,itemsep=1pt]
                \item $C_\ell \cap C_{\ell'} \cap \closure{C}_{\ell''} \ne \emptyset$ for all $\ell, \ell', \ell'' \in [\lambda]$ such that $\ell, \ell' \ne \ell''$.
                \item $u_{\vb{x}}(S) = \max_{\ell \in [\lambda]} \{\alpha x_\ell \cdot \One(S \supseteq C_\ell)\}$.
            \end{enumerate}
        \end{enumerate}
        \item Let $S \subseteq [\beta]$ and $\vb{x}, \vb{y} \in \{0,1,2\}^\lambda$ such that $x_\ell = 1$ and $y_\ell = 2$. Then either
        \begin{enumerate}[topsep=0pt,itemsep=1pt]
            \item $u_{\vb{x}}, u_{\vb{y}}$ are GS or XOS valuations, and $u_{\vb{y}}(S) - u_{\vb{x}}(S) \geq \abs{S \cap C_\ell} - 2\abs{S \cap \closure{C}_\ell} - 10\alpha/\eta^2$.
            \item $u_{\vb{x}}, u_{\vb{y}}$ are general valuations, and $u_{\vb{y}}(S) - u_{\vb{x}}(S) \geq \alpha \cdot \One(S \supseteq C_\ell) - 8\abs{S \cap \closure{C}_\ell}$.
        \end{enumerate}
        \item Let there be $n = \kappa\eta$ bidders, divided into $\kappa$ groups $N_1, \dots, N_\kappa$ of size $\eta$.
        \item Sample $\vb{x}_1, \dots, \vb{x}_n \in \{0,1\}^\lambda$ such that independently for each $k \in [\kappa]$, $\ell \in [\lambda]$, there is a uniformly random $i \in N_k$ such that $x_{i,\ell} = 1$, and $x_{i',\ell} = 0$ for all other $i' \in N_k$.
        \item Sample disjoint $A_1, \dots, A_\kappa, B \subseteq M$ uniformly randomly such that $\abs{A_k} = \alpha$ for all $k \in [\kappa]$ and $\abs{B} = \beta-\alpha$.
        \item Sample uniformly random $\vartheta \in [\lambda]$.
        \item For each $k \in [\kappa]$, sample a uniformly random bijection $\varphi'_k : C_\vartheta \to A_k$. Also sample a uniformly random bijection $\closure{\varphi} : \closure{C}_\vartheta \to B$, and define $\varphi_k = \closure{\varphi} \cup \varphi'_k$ as a bijection from $[\beta]$ to $A_k \cup B$. Extend $\varphi_k$ to sets in the natural way.
        \item Denote $A_{k,\ell} = \varphi_k(C_\ell)$ for all $k \in [\kappa]$, $\ell \in [\lambda]$.
        \item Let $v_1, \dots, v_n : 2^M \to \mathbb{Z}_+$ be the bidders' valuations, where for $k \in [\kappa]$, $i \in N_k$, we define $v_i(S) = u_{\vb{x}_i}(\varphi^{-1}_k(S \cap (A_k \cup B)))$.
        \item Each bidder $i \in N_k$ is given $\varphi_k$ (which reveals $A_k \cup B$) and $\vb{x}_i$ as part of their input, and the referee is given $\varphi_1, \dots, \varphi_\kappa$ (which reveals $A_1, \dots, A_\kappa, B$) for free.
    \end{enumerate}
\end{definition}

At a high level, there is a shared pool of items $B$, and separate sets $A_1, \dots, A_\kappa$ associated with each group of bidders $N_1, \dots, N_\kappa$. Nearly all the welfare comes from the items $A_1, \dots, A_\kappa$, but there is only one bidder $i \in N_k$ who has high value $\alpha$ for $A_k$. But the bidders need to tell the referee who has high value for $A_k$, without knowing which of the sets $C_1, \dots, C_\lambda$ has been mapped onto $A_k$. Thus, each group of bidders $N_k$ is essentially playing the $\abs{N_k}$-ary INDEX problem with the referee, which is hard in the one-way communication model. Properties (3) and (4) have two branches since the non-subadditivity of general valuations allows for a stronger construction, but the key ideas remain the same.

\subsection{Parameters for Each Valuation Class}

We now establish what KSSHDs are possible for different valuation classes. For simplicity, in the three following propositions, let $\gamma$ be chosen such that $\alpha \leq \beta$ and $\eta \geq 2$, which ensures we can take $\alpha$-sized subsets of $[\beta]$, and that the total number of items in the KSSHD does not exceed $m$. This is achieved by letting $m$ be sufficiently large and $\gamma \in [\log m, m/10^{10}]$.

\begin{proposition} \label{prop:GSKSSHD}
    For $\gamma \in [\log m, m/10^{10}]$, there exists an $(\alpha, \eta, \lambda, \vb{C}, \calU)$-KSSHD using GS valuations for parameters $\alpha = (m/\gamma)^{2/7} \gamma$, $\eta = (m/\gamma)^{1/7}$, $\lambda = 2^\gamma$.
\end{proposition}
\begin{proof}
The construction of~\cite{BalcanH11} (\Cref{thm:BalcanHarvey}) with $d = \alpha$, $b = 10\gamma = 10\alpha/\eta^2$, $\beta = m/\eta = 100d^3/b^2$, and $\lambda = 2^{b/10}$ yields matroid rank functions $r_{\vb{x}} : 2^{[\beta]} \to \mathbb{Z}_+$ such that
\begin{itemize}[topsep=1pt,itemsep=1pt]
    \item $r_{\vb{x}}([\beta]) \leq d$.
    \item $r_{\vb{x}}(C_\ell) = \alpha$ if $x_\ell = 1$.
    \item $r_{\vb{x}}(C_\ell) = 10\alpha/\eta^2$ if $x_\ell = 0$.
\end{itemize}

Redefine $r_{\vb{0}}$ to be the zero function. For $\vb{x} \in \{0,1,2\}^\lambda$, let $\vb{x}^-, \vb{x}^+ \in \{0, 1\}^\lambda$ be defined by $x^-_\ell = \One(x_\ell \geq 1)$ and $x^+_\ell = \One(x_\ell \geq 2)$. Define $u_{\vb{x}} = r_{\vb{x}^-} \star (2 \cdot r_{\vb{x}^+})$. As the convolution of two weighted matroid rank functions, $u_{\vb{x}}$ is GS.

Property (3.a.i) holds because the convolved functions are weighted matroid rank functions with item weights at most $2$, and the maximum values of $r_{\vb{x}^-}$ and $2 \cdot r_{\vb{x}^+}$ are $\alpha$ and $2\alpha$, respectively. Property (3.a.ii) holds by construction. Property (3.a.iii) holds because the contribution of $r_{\vb{x}^-}$ is at most $10\alpha/\eta^2$ and the contribution of $2 \cdot r_{\vb{x}^+}$ is at most $20\alpha/\eta^2$.

Property (4.a) holds because except for at most $10\alpha/\eta^2$ items, $u_{\vb{x}}$ obtains at most marginal value $1$ per item of $C_\ell$, and at most marginal value $2$ everywhere else. On the other hand, $u_{\vb{y}}$ obtains marginal value $2$ per item of $C_\ell$. Taking the difference yields the desired result.
\end{proof}

\begin{proposition} \label{prop:XOSKSSHD}
    For $\gamma \in [\log m, m/10^{10}]$, there exists an $(\alpha, \eta, \lambda, \vb{C}, \calU)$-KSSHD using XOS valuations for parameters $\alpha = (m/\gamma)^{2/5} \gamma$, $\eta = (m/\gamma)^{1/5}$, $\lambda = 2^\gamma$.
\end{proposition}
\begin{proof}
Let $\beta = m/\eta = \alpha^2/\gamma$. Sample a uniformly random list of sets $C_1, \dots, C_\lambda \subseteq [\beta]$ of size $\alpha$. For any distinct $\ell, \ell' \in [\lambda]$, we have by a Chernoff bound that
\begin{align*}
    \Pr[\abs{C_\ell \cap C_{\ell'}} \geq 10\alpha/\eta^2] \quad &= \quad \Pr[\abs{C_\ell \cap C_{\ell'}} \geq 10\gamma] \\
    &= \quad \Pr[\abs{C_\ell \cap C_{\ell'}} \geq 10\E[\abs{C_\ell \cap C_{\ell'}}]] \quad \leq \quad e^{-2\gamma} \enspace .
\end{align*}
By a union bound, the probability that the event above holds for any distinct $\ell, \ell' \in [\lambda]$ is at most $\lambda^2 e^{-2\gamma} < 1$. Therefore, fixing any choice of $C_1, \dots, C_\lambda$ such that the low intersection property holds, the valuation $u_{\vb{x}} : 2^{[\beta]} \to \mathbb{Z}_+$ defined by $u_{\vb{x}}(S) = \max_{\ell \in [\lambda]} x_\ell\abs{S \cap C_\ell}$ satisfies Property (3.a).

Property (4.a) holds because except for at most $10\alpha/\eta^2$ items, $u_{\vb{x}}$ obtains at most marginal value $1$ per item of $C_\ell$, and at most marginal value $2$ everywhere else. On the other hand, $u_{\vb{y}}$ obtains marginal value $2$ per item of $C_\ell$. Taking the difference yields the desired result.
\end{proof}

\begin{proposition} \label{prop:GeneralKSSHD}
    For $\gamma \in [\log m, m/10^{10}]$, there exists an $(\alpha, \eta, \lambda, \vb{C}, \calU)$-KSSHD using general valuations for parameters $\alpha = 100\gamma$, $\eta = m/(200\gamma)$, $\lambda = 2^\gamma$.
\end{proposition}
\begin{proof}
Let $\beta = m/\eta = 200\gamma = 2\alpha$. Sample a uniformly random list of sets $C_1, \dots, C_\lambda \subseteq [200\gamma]$ of size $100\gamma$. For $\ell, \ell', \ell'' \in [\lambda]$ such that $\ell, \ell' \ne \ell''$, the set $C_\ell \cap C_{\ell'} \cap \closure{C}_{\ell''}$ is the intersection of either two or three independent random sets of size $100\gamma$.

If $\ell = \ell'$, we have $\abs{C_{\ell} \cap C_{\ell'}} = 100\gamma$. Otherwise, by a Chernoff bound, $\abs{C_\ell \cap C_{\ell'}} \in [25\gamma, 75\gamma]$ w.p.\ at least $1 - 2e^{-3\gamma}$ when $\ell \ne \ell'$. Conditioned on this event, another application yields that $\abs{C_\ell \cap C_{\ell'} \cap \closure{C}_{\ell''}} = 0$ w.p.\ at most $e^{-3\gamma}$. Thus, by a union bound, the probability that $C_\ell \cap C_{\ell'} \cap \closure{C}_{\ell''} = \emptyset$ for any $\ell, \ell', \ell'' \in [\lambda]$ such that $\ell, \ell' \ne \ell''$ is at most $3\lambda^3 e^{-3\gamma} < 1$. Therefore, fix any choice of $C_1, \dots, C_\lambda$ such that Property (3.b.i) holds, and define $u_{\vb{x}}$ according to Property (3.b.ii).

Now, consider $\vb{x}, \vb{y} \in \{0,1,2\}^\lambda$ such that $x_\ell = 1$ and $y_\ell = 2$. Then
\[
    u_{\vb{y}}(S) - u_{\vb{x}}(S) \quad \geq \quad \alpha \cdot \One(S \supseteq C_\ell) - 2\alpha \cdot \max_{\ell' \in [\lambda] \setminus \{\ell\}} \One(S \supseteq C_{\ell'}) \enspace .
\]
By our choice of $C_1, \dots, C_\lambda$, we have $\abs{C_{\ell'} \cap \closure{C}_\ell} \geq 25\gamma \geq \alpha/4$ for all $\ell' \ne \ell$, and so 
\[
    u_{\vb{y}}(S) - u_{\vb{x}}(S) \quad \geq \quad \alpha \cdot \One(S \supseteq C_\ell) - 8\abs{S \cap \closure{C}_\ell} \enspace ,
\]
so Property (4.b) is satisfied.
\end{proof}

\subsection{Simultaneous Communication Lower Bound}

Now, we prove communication lower bounds on simultaneous algorithms with good approximations on KSSHDs. For $k \in [\kappa]$, let $i_k \in N_k$ denote the bidder with high value for $A_k$. First, we prove that good approximations require allocating a significant fraction of $A_k$ to bidder $i_k$.

\begin{lemma} \label{lemma:MustAllocateToCorrectBiddersXOS}
    For any allocation $(S_1, \dots, S_n)$ on an instance drawn from an $(\alpha, \eta, \lambda, \vb{C}, \calU)$-KSSHD for GS or XOS valuations,
    \[
        \sum_{i \in [n]} v_i(S_i) \quad \leq \quad 17m/\eta + 2\sum_{k \in [\kappa]} \abs{S_{i_k} \cap A_k} \enspace .
    \]
\end{lemma}
\begin{proof}
By subadditivity and Property (3.a),
\begin{align*}
    \sum_{i \in [n]} v_i(S_i) \quad &\leq \quad \sum_{k \in [\kappa]} \sum_{i \in N_k} \Big(v_i(S_i \cap A_k) + v_i(S_i \cap B)\Big) \\
    &\leq \quad 2\abs{B} + \sum_{k \in [\kappa]} \sum_{i \in N_k} v_i(S_i \cap A_k) \\
    &\leq \quad 2\beta + \sum_{k \in [\kappa]} \bigg(v_{i_k}(S_{i_k} \cap A_k) + \sum_{i \in N_k : i \ne i_k} v_i(A_k)\bigg) \\
    &\leq \quad 2\beta + 30n\alpha/\eta^2 + 2\sum_{k \in [\kappa]} \abs{S_{i_k} \cap A_k} \\
    &\leq \quad 17m/\eta + 2\sum_{k \in [\kappa]} \abs{S_{i_k} \cap A_k} \enspace . \qedhere
\end{align*}
\end{proof}

\begin{lemma} \label{lemma:MustAllocateToCorrectBiddersGEN}
    For any allocation $(S_1, \dots, S_n)$ on an instance drawn from an $(\alpha, \eta, \lambda, \vb{C}, \calU)$-KSSHD for general valuations,
    \[
        \sum_{i \in [n]} v_i(S_i) \quad \leq \quad 2m/\eta + 2\alpha\sum_{k \in [\kappa]} \One(S_{i_k} \supseteq A_k) \quad \leq \quad 2m/\eta + 2\sum_{k \in [\kappa]} \abs{S_{i_k} \cap A_k} \enspace .
    \]
\end{lemma}
\begin{proof}
By Property (9), the items $\closure{C}_\vartheta$ are mapped to $B$ in the same way for all bidders. Therefore, Property (3.b.i) says that for $k, k' \in [\kappa]$, $\ell, \ell' \in [\lambda]$ such that $\ell, \ell' \ne \vartheta$,
\[
    A_{k,\ell} \cap A_{k',\ell'} \quad = \quad \varphi_k(C_\ell) \cap \varphi_{k'}(C_{\ell'}) \quad \supseteq \quad \varphi_k(C_\ell \cap C_{\ell'} \cap \closure{C}_\vartheta) \quad \ne \quad \emptyset \enspace .
\]
In plain English, this means that the only pairs of disjoint sets are pairs where at least one member is in $A_1, \dots, A_\kappa$. Since Property (3.b.ii) gives bidders ``all or nothing'' valuations, this implies that at most one bidder can get positive value (which is at most $\alpha \leq \beta = m/\eta$) from a set which is not one of $A_1, \dots, A_\kappa$. Thus, since $\alpha \cdot \One(S_{i_k} \supseteq A_k) \leq \abs{S_{i_k} \cap A_k}$,
\[
    \sum_{i \in [n]} v_i(S_i) \quad \leq \quad 2m/\eta + 2\alpha\sum_{k \in [\kappa]} \One(S_{i_k} \supseteq A_k) \quad \leq \quad 2m/\eta + 2\sum_{k \in [\kappa]} \abs{S_{i_k} \cap A_k} \enspace . \qedhere
\]
\end{proof}

\begin{theorem} \label{thm:KSSHDIsHard}
    Any simultaneous algorithm attaining an $\eta/40$-approximation on an $(\alpha, \eta, \lambda, \vb{C}, \calU)$-KSSHD uses at least $\lambda/\eta^2$ communication.
\end{theorem}
\begin{proof}
Suppose for contradiction that an algorithm $\ALG'$ attaining an $\eta/40$-approximation with less than $\lambda/\eta^2$ communication exists. Then fix the randomness of $\ALG'$ to the outcome that attains the best approximation in expectation over the KSSHD, and denote this deterministic algorithm (which also attains an $\eta/40$-approximation in expectation over the KSSHD) by $\ALG$. Let $(S_1, \dots, S_n)$ be the allocation output by $\ALG$.

The randomized inputs are the index $\vartheta$, the mappings $\varphi_1, \dots, \varphi_\kappa$, the sets $A_1, \dots, A_\kappa, B$, and the vectors $\vb{x}_1, \dots, \vb{x}_n$. For any $k \in [\kappa]$ and bidder $i \in N_k$, observe the following:
\begin{itemize}[topsep=1pt,itemsep=1pt]
    \item Bidder $i$ sees $\varphi_k$ and $\vb{x}_i$, which fully determines their input.
    \item With no other information, $\varphi_k$ looks like a uniformly random bijection between $[\beta]$ and $A_k \cup B$, which looks like a uniformly random set of size $\beta$.
    \item The referee already knows $\varphi_1, \dots, \varphi_\kappa$.
\end{itemize}
Therefore, for all $k \in [\kappa]$, let $\vb{j}_k$ be the random vector where $j_{k,\ell} \in N_k$ is the bidder for which $x_{j_{k,\ell},\ell} = 1$. Since $\vb{j}_k$ is the only input unknown to the referee, the messages of the bidders in group $N_k$, conditional on $\varphi_k$, are purely a function of $\vb{j}_k$.

We now show that if the messages are only of length $\lambda/\eta^2$, the expected conditional distribution of $j_{k,\vartheta}$ (which is the bidder $i_k \in N_k$ with high value for $A_k$) is not far from uniform, because $\vartheta$ is a uniformly random index in $[\lambda]$.

Let $k \in [\kappa]$. Let $\mu_k(\vb{j}_k)$ be the concatenation of the messages of every bidder in $N_k$.

\begin{lemma} \label{lemma:LowDistributionShift}
    For all $k \in [\kappa]$,
    \[
        \E[\TV(\Law(j_{k,\vartheta} \given \mu_k(\vb{j}_k), \varphi_k, \vartheta), \Unif(N_k))] \quad \leq \quad \eta^{-1} \enspace .
    \]
\end{lemma}
\begin{proof}
By Jensen's Inequality and the fact that $\varphi_k, \vartheta$ are independent from $\vb{j}_k$,
\begin{align*}
    \text{LHS}^2 \quad &\leq \quad \E[\TV(\Law(j_{k,\vartheta} \given \mu_k(\vb{j}_k), \varphi_k, \vartheta), \Unif(N_k))^2] \\
    &= \quad \E[\TV(\Law(j_{k,\vartheta} \given \mu_k(\vb{j}_k), \varphi_k, \vartheta), \Law(j_{k,\vartheta} \given \varphi_k, \vartheta))^2] \\
    &= \quad \frac{1}{\lambda} \sum_{\ell \in [\lambda]} \E[\TV(\Law(j_{k,\ell} \given \mu_k(\vb{j}_k), \varphi_k), \Law(j_{k,\ell} \given \varphi_k))^2] \\
    &\leq \quad \frac{1}{\lambda} \sum_{\ell \in [\lambda]} \I(j_{k,\ell} \mutual \mu_k(\vb{j}_k) \given \varphi_k) \\
    &\leq \quad \frac{1}{\lambda} \I(\vb{j}_k \mutual \mu_k(\vb{j}_k) \given \varphi_k) \\
    &\leq \quad \frac{1}{\lambda} \Ent(\mu_k(\vb{j}_k) \given \varphi_k) \\
    &< \quad \frac{1}{\eta^2} \enspace ,
\end{align*}
where the fourth line is by Pinsker's Inequality, the fifth line is by the superadditivity of mutual information for independent random variables, and the last line is because there are less than $2^{\lambda/\eta^2}$ messages of length less than $\lambda/\eta^2$, and the entropy of a random variable is at most the logarithm of its support. Taking square roots completes the proof.
\end{proof}

Fix some $k \in [\kappa]$, and let $\zeta_k$ denote the information other than $\mu_k(\vb{j}_k), \varphi_k$ available to the referee: the messages of the bidders not in $N_k$ and the other mappings. For $i \in N_k$, let $S_i(\vb{j}_k, \zeta_k)$ be the set allocated to bidder $i$ when the referee has the information $\mu_k(\vb{j}_k), \zeta_k$. Then
\begin{align*}
    \E[\abs{S_{i_k} \cap A_k}] \quad &= \quad \E\Big[\E[\abs{S_{j_{k,\vartheta}}(\vb{j}_k, \zeta_k) \cap A_k} \given \mu_k(\vb{j}_k), \zeta_k]\Big] \\
    &= \quad \E\bigg[\sum_{i \in N_k} \Big(\Pr[j_{k,\vartheta} = i \given \mu_k(\vb{j}_k), \varphi_k, \vartheta, \zeta_k] \cdot \abs{S_i(\vb{j}_k, \zeta_k) \cap A_k}\Big)\bigg] \\
    &\leq \quad \alpha \cdot \E\Big[\max_{i \in N_k} \Pr[j_{k,\vartheta} = i \given \mu_k(\vb{j}_k), \varphi_k, \vartheta, \zeta_k]\Big] \\
    &= \quad \alpha \cdot \E\Big[\max_{i \in N_k} \Pr[j_{k,\vartheta} = i \given \mu_k(\vb{j}_k), \varphi_k, \vartheta]\Big] \\
    &\leq \quad \alpha \cdot \E\Big[1/\eta + \TV(\Law(j_{k,\vartheta} \given \mu_k(\vb{j}_k), \varphi_k, \vartheta), \Unif(N_k))\Big] \\
    &< \quad \frac{2\alpha}{\eta} \enspace ,
\end{align*}
where the fourth line is because $j_{k,\vartheta}$ conditioned on $\mu_k(\vb{j}_k), \varphi_k, \vartheta$ is independent from $\zeta_k$ (the other messages and mappings), the fifth line is by definition of total variation distance, and the sixth line is by~\Cref{lemma:LowDistributionShift}. To conclude, we sum over $k \in [\kappa]$ to get
\[
    \E\bigg[\sum_{k \in [\kappa]} \abs{S_{i_k} \cap A_k}\bigg] \quad < \quad 2\alpha\kappa/\eta \quad = \quad m/\eta \enspace .
\]
By~\Cref{lemma:MustAllocateToCorrectBiddersXOS,lemma:MustAllocateToCorrectBiddersGEN}, the expected welfare obtained by $\ALG$ is less than $19m/\eta$. But since the optimal welfare is at least $\alpha\kappa = m/2$ (allocate $A_k$ to bidder $i_k$ for all $k \in [\kappa]$), this contradicts the assumption that $\ALG$ obtains an $\eta/40$-approximation in expectation over the KSSHD.
\end{proof}

\section{DSIC Hardness of Approximation}

In this section, we show how to convert a deterministic DSIC mechanism into a good simultaneous algorithm for KSSHDs. We first define a suitable domain for our DSIC mechanisms to operate on.

\subsection{The Valuation Class}

Add a special item $\dagger$, so that the total item set consists of the regular items $[m]$ (on which the KSSHD is defined) together with $\dagger$. We will denote $M = [m] \cup \{\dagger\}$. This special item will primarily be used by GS and XOS valuations.

Let $\calH$ denote an $(\alpha, \eta, \lambda, \vb{C}, \calU)$-KSSHD using either GS or XOS valuations, and let $\calV(\calH)$ be the class of valuations consisting of every item-relabeled $u \in \calU$, along with the zero valuation. For any $v \in \calV(\calH)$, $s \in \mathbb{R}$, and $t \in \mathbb{R}_+$, define $v_{s,t}$ by
\[
    v_{s,t}(S) \quad = \quad 2^{s/2} \Big(v(S \cap [m]) + (\alpha\eta + t/m) \cdot \One(\dagger \in S)\Big) \enspace .
\]
Observe that if $v$ is GS or XOS, $v_{s,t}$ remains so. We may think of $s$ as a scaling factor, and $t$ as a special item weight. Define
\[
    \calV[s^-, s^+](\calH) \quad = \quad \{v_{s,t} : s \in \mathbb{Z} \cap [s^-, s^+],\, t \in \mathbb{Z} \cap [0, m^4],\, v \in \calV(\calH)\} \enspace .
\]

Now, let $\calH$ denote an $(\alpha, \eta, \lambda, \vb{C}, \calU)$-KSSHD using general valuations, and let $\calV(\calH)$ be the class of valuations consisting of every item-relabeled $u \in \calU$, along with the zero valuation. For any $v \in \calV(\calH)$, $s \in \mathbb{R}$, $t \in \mathbb{R}_+$, and $X, Y \in 2^M \setminus \{\emptyset\}$, define $v_{s,t,X,Y}$ by
\[
    v_{s,t,X,Y}(S) \quad = \quad 2^{s/2} \Big(v(S \cap [m]) \cdot \One(S \not\supseteq X \land S \not\supseteq Y) + (2\alpha + t/m) \cdot \One(S \supseteq X \lor S \supseteq Y)\Big) \enspace .
\]
Observe that $v(S \cap [m]) \leq 2\alpha$, so $v_{s,t,X,Y}$ is monotone. Abbreviate $v_{s,t} = v_{s,t,M,M}$. Define
\[
    \calV[s^-, s^+](\calH) \quad = \quad \{v_{s,t,X,Y} : s \in \mathbb{Z} \cap [s^-, s^+],\, t \in \mathbb{Z} \cap [0, m^4],\, X, Y \in 2^M \setminus \{\emptyset\},\, v \in \calV(\calH)\} \enspace .
\]

We remark that for $s \in [-2^m, 2^m]$, the values for these valuations can be represented in $O(m)$ bits.

\subsection{Finding an Informative Node}

This section proves the following proposition, which states that a good DSIC mechanism requires (almost) every bidder to signal something about their value for the special outcome in the first round of non-trivial communication.

\begin{proposition} \label{prop:InformativeNode}
    Let $\eta \geq 800$ and $\alpha\eta \leq m$, and let $\calH$ denote an $(\alpha, \eta, \lambda, \vb{C}, \calU)$-KSSHD using GS, XOS, or general valuations. Suppose $\MECH'$ is a deterministic, IR, NNT, DSIC mechanism which attains an $\eta/800$-approximation using $\Gamma \leq 2^{m/2}$ communication on $\calV[-2^m, 2^m](\calH)$. Then there exists a mechanism $\MECH$ such that
    \begin{enumerate}[topsep=1pt,itemsep=1pt]
        \item $\MECH$ is a deterministic, IR, NNT, DSIC mechanism which attains an $\eta/800$-approximation using $\Gamma$ communication on $\calV[-m, m](\calH)$.
        \item Let $\xi_i(v)$ denote the message $\MECH$ prescribes bidder $i$ to send when their valuation is $v$. Then there exists a bidder $i^\dagger \in [n]$, such that for all $i \ne i^\dagger$ and valuations $v \in \calV(\calH)$, there exists $t \in [m^4]$ such that $\xi_i(v_{0,0}) = \xi_i(v_{0,t-1}) \ne \xi_i(v_{0,t})$.
    \end{enumerate}
\end{proposition}
\begin{proof}
We first find the special node in $\MECH'$ where the first non-trivial communication occurs. For the proof of~\Cref{lemma:FindingTheNode}, we abuse notation for general valuations and let $v_{s,t}$ denote arbitrary $v_{s,t,X,Y}$ when the exact values of $X, Y$ do not matter.

\begin{lemma} \label{lemma:FindingTheNode}
    There exists a node $\calN^\dagger$ of the communication protocol of $\MECH'$, a bidder $i^\dagger \in [n]$, $u, u' \in \calV(\calH)$, $a, a', s^\dagger \in [-2^m, 2^m]$, and $b, b' \in [m^4] \cup \{0\}$, such that
    \begin{enumerate}[topsep=1pt,itemsep=1pt]
        \item At $\calN^\dagger$, the messages sent by the prescribed strategy of bidder $i^\dagger$ for the valuations $u_{a,b}$ and $u'_{a',b'}$ are different.
        \item Either:
        \begin{itemize}[topsep=0pt,itemsep=0pt]
            \item $u_{a,b}(\{\dagger\}), u'_{a',b'}(\{\dagger\}) \in [2^{s^\dagger/2} \alpha\eta^2, 2^{s^\dagger/2+1} \alpha\eta^2]$ and the valuations are GS or XOS.
            \item $u_{a,b}(M), u'_{a',b'}(M) \in [2^{s^\dagger/2} \alpha\eta^2, 2^{s^\dagger/2+1} \alpha\eta^2]$ and the valuations are general.
        \end{itemize}
        \item $a, a' \in [s^\dagger - m, s^\dagger + m]$.
        \item Every profile of valuations from $\calV[s^\dagger-m, s^\dagger+m](\calH)$ reaches $\calN^\dagger$.
    \end{enumerate}
\end{lemma}
\begin{proof}
Initialize $s^- = -2^m$ and $s^+ = 2^m$, and initialize $\calN$ to be the root node of the communication protocol. Let $\xi_{\calN,i}(v)$ denote the message $\MECH'$ prescribes bidder $i$ to send at node $\calN$ when their valuation is $v$. We run the following procedure:
\begin{enumerate}[topsep=1pt,itemsep=1pt]
    \item If $\calN$ is a leaf, stop.
    \item If $\calN$ has no communication, set $\calN$ to the fixed next node of the protocol and repeat.
    \item Otherwise, if there is a bidder $i$ and $v_{s,t}, v'_{s',t'} \in \calV[s^- + 2m, s^+ - 2m]$ such that $\xi_{\calN,i}(v_{s,t}) \ne \xi_{\calN,i}(v'_{s',t'})$, stop.
    \item Otherwise, for every profile of valuations from $\calV[s^- + 2m, s^+ - 2m](\calH)$, the next node of the protocol is some fixed node $\calN'$. Set $s^- \leftarrow s^- + 2m$, $s^+ \leftarrow s^+ - 2m$, and $\calN \leftarrow \calN'$ and repeat.
\end{enumerate}

Let $\calN, s^-, s^+$ be the state when the procedure stops. Since $s^+ - s^-$ only decreases by $4m$ at nodes with nonzero communication, we have $s^+ - s^- \geq 2^{m+1} - 4m\Gamma \geq 0$, so $[s^-, s^+]$ is non-empty.

Now, let $o \in \calV(\calH)$ be the zero valuation, and observe that $o_{s^-,0}, o_{s^-,m^4} \in \calV[s^-, s^+](\calH)$ (in the case of general valuations, consider $o_{s^-,0,\{\dagger\},\{\dagger\}}, o_{s^-,m^4,\{\dagger\},\{\dagger\}}$). Since every profile of valuations from $\calV[s^-, s^+](\calH)$ reaches $\calN$, the profiles where one bidder has $o_{s^-,m^4}$ and everyone else has $o_{s^-,0}$ reach $\calN$, and the approximation guarantee forces $\MECH'$ to allocate $\dagger$ to the bidder with valuation $o_{s^-,m^4}$. Since any bidder could have $o_{s^-,m^4}$, $\calN$ cannot be a leaf.

Thus, the procedure must stop in Step (3) with a bidder $i$ and valuations $v_{s,t}, v'_{s',t'} \in \calV[s^- + 2m, s^+ - 2m]$. Let $\Delta = \sign(s' - s)$ and let $\delta = \sign(t' - t)$. Consider the chain of valuations
\[
    v_{s,t} \quad v_{s+\Delta,t} \quad \dots \quad v_{s'-\Delta,t} \quad v_{s',t} \quad v_{s',t+\delta} \quad \dots \quad v_{s',t'-\delta} \quad v_{s',t'} \quad v'_{s',t'} \enspace ,
\]
where we change from one valuation to the other by incrementing or decrementing $s$ and $t$. Since $\xi_{\calN,i}(v_{s,t}) \ne \xi_{\calN,i}(v'_{s',t'})$, there is some adjacent pair $u_{a,b}, u'_{a',b'}$ such that $\xi_{\calN,i}(u_{a,b}) \ne \xi_{\calN,i}(u'_{a',b'})$. These will be the valuations satisfying Property (1).

By construction, $\abs{a - a'} + \abs{b - b'} \leq 1$, so for GS or XOS valuations, we have
\[
    \abs{\log(u_{a,b}(\{\dagger\})) - \log(u'_{a',b'}(\{\dagger\}))} \quad \leq \quad \abs{\log(2^{a/2}(\alpha\eta + b/m)) - \log(2^{a'/2}(\alpha\eta + b'/m))} \quad \leq \quad 1/2 \enspace ,
\]
and for general valuations, we have
\[
    \abs{\log(u_{a,b}(M)) - \log(u'_{a',b'}(M))} \quad \leq \quad \abs{\log(2^{a/2}(2\alpha + b/m)) - \log(2^{a'/2}(2\alpha + b'/m))} \quad \leq \quad 1/2 \enspace .
\]
On a logarithmic scale, the intervals $[2^{s/2} \alpha\eta^2, 2^{s/2+1} \alpha\eta^2]$ have width $1$ and appear every increment of $1/2$, and thus there exists some $s^\dagger$ satisfying Property (2) for $u_{a,b}, u'_{a',b'}$.

Since $b/m, b'/m \leq m^3$, we have $a, a' \in [s^\dagger - 2\log(m^6), s^\dagger + 2\log(m^6)] \subseteq [s^\dagger - m, s^\dagger + m]$, which yields Property (3). Finally, since $a, a' \in [s^- + 2m, s^+ - 2m]$, we have $[s^\dagger - m, s^\dagger + m] \subseteq [s^-, s^+]$, and by construction, every profile of valuations from $\calV[s^-, s^+]$ reaches $\calN$, yielding Property (4).
\end{proof}

We now repackage the previous lemma into a cleaner form: instead of dealing with ``the first non-trivial communication'' and ``some scale $s^\dagger$,'' we construct a functionally identical mechanism where everything interesting happens in the actual first round and at the default scale.

\begin{lemma} \label{lemma:MinimalMechanism}
    There exists a deterministic, IR, NNT, DSIC mechanism $\MECH$ which attains an $\eta/800$-approximation using $\Gamma$ communication on $\calV[-m, m](\calH)$. This mechanism has a bidder $i^\dagger \in [n]$, and $u, u' \in \calV[-m, m](\calH)$ such that
    \begin{enumerate}[topsep=1pt,itemsep=1pt]
        \item In the first round of communication, the messages sent by the prescribed strategy of bidder $i^\dagger$ for the valuations $u, u'$ are different.
        \item Either:
        \begin{itemize}[topsep=0pt,itemsep=0pt]
            \item $u(\{\dagger\}), u'(\{\dagger\}) \in [\alpha\eta^2, 2\alpha\eta^2]$ and the valuations are GS or XOS.
            \item $u(M), u'(M) \in [\alpha\eta^2, 2\alpha\eta^2]$ and the valuations are general.
        \end{itemize}
    \end{enumerate}
\end{lemma}
\begin{proof}
Let $\sigma'_1, \dots, \sigma'_n$ be the strategy prescriptions of $\MECH'$, and let $i^\dagger, \calN^\dagger$ be the bidder and node promised by~\Cref{lemma:FindingTheNode}.

Let $\MECH$ be the mechanism whose communication protocol is the tree rooted at $\calN^\dagger$, and which prescribes the strategies $\sigma_1, \dots, \sigma_n$, where for all $v_{s,t} \in \calV[-m, m](\calH)$, we define $\sigma_i(v_{s,t}) = \sigma'_i(v_{s+s^\dagger,t})$. If running $\MECH'$ starting from $\calN^\dagger$ yields the allocation $(S_1, \dots, S_n)$ and payments $(p_1, \dots, p_n)$, then $\MECH$ outputs $(S_1, \dots, S_n)$ and $(2^{-s^\dagger/2} p_1, \dots, 2^{-s^\dagger/2} p_n)$. In other words, $\MECH$ asks each bidder to scale their valuation by $2^{s^\dagger/2}$, then runs $\MECH'$ starting from $\calN^\dagger$.

Observe that every bidder $i$ plays the strategy $\sigma'_i(v)$ for a valuation $v \in \calV[s^\dagger - m, s^\dagger + m](\calH)$. Therefore, by Properties (3) and (4) of~\Cref{lemma:FindingTheNode}, every profile of valuations reaches $\calN^\dagger$, and hence $\MECH$ realizes the same outcomes as $\MECH'$ on the relevant domain, using at most the same amount of communication.

Determinism, IR, NNT, DSIC, and welfare approximation are preserved under scaling, so $\MECH$ satisfies all the same properties on the restricted domain as $\MECH'$. Properties (1) and (2) of~\Cref{lemma:FindingTheNode} then directly complete the proof.
\end{proof}

Fix the mechanism $\MECH$, and let bidder $i^\dagger$ and valuations $u, u' \in \calV[-m, m](\calH)$ be the parameters promised by~\Cref{lemma:MinimalMechanism}. Let $\xi_1, \dots, \xi_n$ be the messages sent in the first round by the strategies prescribed by $\MECH$. We show that because $u, u'$ are similar valuations that send different messages, the other bidders have adversarial strategies to deny bidder $i^\dagger$ utility (violating DSIC), unless every other bidder is forced to \emph{commit} some signal for how much they value the special outcome.

We remind the reader that for general valuations, we once again denote $v_{s,t} = v_{s,t,M,M}$.

\begin{lemma} \label{lemma:MustSignalSpecialOutcome}
    For all bidders $i \ne i^\dagger$ and all valuations $v \in \calV(\calH)$, $\xi_i(v_{0,0}) \ne \xi_i(v_{0,m\alpha\eta^3})$.
\end{lemma}
\begin{proof}
Suppose for contradiction that the claim is false. Let $o \in \calV(\calH)$ be the zero valuation.

First, suppose that we are dealing with GS or XOS valuations.
\begin{itemize}[topsep=1pt,itemsep=1pt]
    \item Consider the two profiles where bidder $i$ has $v_{0,0}$, bidder $i^\dagger$ has $u$ or $u'$, and everyone else has $o_{0,0}$. Then since $u(\{\dagger\}), u'(\{\dagger\}) \geq \alpha\eta^2$, $\MECH$ must get welfare at least $800\alpha\eta$ when all bidders play truthfully. But since $u([m]) \leq u(\{\dagger\})/(\eta/3) \leq 6\alpha\eta$ (and similarly for $u'$), $v_{0,0}([m]) \leq 3\alpha$, $o_{0,0}([m]) = 0$, and every bidder other than $i^\dagger$ has value $\alpha\eta$ for $\dagger$, $\MECH$ must allocate $\dagger$ to $i^\dagger$ to satisfy the approximation guarantee. By IR, the payment of $i^\dagger$ on this profile under truthful play is at most $u(M) \leq u([m]) + u(\{\dagger\}) \leq 3\alpha\eta^2$.

    \item Since $\xi_{i^\dagger}(u) \ne \xi_{i^\dagger}(u')$, let $\xi_{i^\dagger}(u) \ne \xi_{i^\dagger}(o_{0,10m\alpha\eta^2})$ WLOG.
    
    \item Now, consider the profile where bidder $i$ has $v_{0,m\alpha\eta^3}$, bidder $i^\dagger$ has $o_{0,10m\alpha\eta^2}$, and everyone else has $o_{0,0}$. Then since $v_{0,m\alpha\eta^3}(\{\dagger\}) \geq \alpha\eta^3$, $\MECH$ must get welfare at least $800\alpha\eta^2$ when all bidders play truthfully. But because every bidder other than $i$ has value $0$ for $[m]$ and value at most $11\alpha\eta^2$ for $\dagger$, $\MECH$ must allocate $\dagger$ to $i$ to satisfy the approximation guarantee. Then by NNT, the utility of $i^\dagger$ on this profile under truthful play must be $0$.

    \item Therefore, consider the case where bidders $j \ne i, i^\dagger$ play $\sigma_j(o_{0,0})$, and bidder $i$ plays $\sigma_i(v_{0,0})$ if bidder $i^\dagger$ sends the message $\xi_{i^\dagger}(u)$ in the first round, and $\sigma_i(v_{0,m\alpha\eta^3})$ otherwise. This is a valid strategy because we assumed for contradiction that $\xi_i(v_{0,0}) = \xi_i(v_{0,m\alpha\eta^3})$. Then if bidder $i^\dagger$ has the valuation $o_{0,10m\alpha\eta^2}$, they send the message $\xi_{i^\dagger}(o_{0,10m\alpha\eta^2}) \ne \xi_{i^\dagger}(u)$ in the first round, so truthful play yields utility $0$. On the other hand, playing $\sigma_{i^\dagger}(u)$ obtains utility at least $10\alpha\eta^2 - 3\alpha\eta^2 = 7\alpha\eta^2$, contradicting DSIC.
\end{itemize}

Now, suppose that we are dealing with general valuations.
\begin{itemize}[topsep=1pt,itemsep=1pt]
    \item Consider the two profiles where bidder $i$ has $v_{0,0}$, bidder $i^\dagger$ has $u$ or $u'$, and everyone else has $o_{0,0}$. Then since $u(M), u'(M) \geq \alpha\eta^2$, $u(\emptyset), u'(\emptyset) = 0$, and the combined value of every other bidder is at most $2\alpha$, $\MECH$ must give non-empty sets $S, S' \subseteq M$ to $i^\dagger$ to satisfy the approximation guarantee when $i^\dagger$ has $u, u'$. By IR, the payment of $i^\dagger$ on these profiles under truthful play is at most $\max\{u(M), u'(M)\} \leq 2\alpha\eta^2$.

    \item Since $\xi_{i^\dagger}(u) \ne \xi_{i^\dagger}(u')$, let $\xi_{i^\dagger}(u) \ne \xi_{i^\dagger}(o_{0,10m\alpha\eta^2,S,S'})$ WLOG.
    
    \item Now, consider the profile where bidder $i$ has $v_{0,m\alpha\eta^3}$, bidder $i^\dagger$ has $o_{0,10m\alpha\eta^2,S,S'}$, and everyone else has $o_{0,0}$. Then since $v_{0,m\alpha\eta^3}(M) \geq \alpha\eta^3$, $\MECH$ must get welfare at least $800\alpha\eta^2$ when all bidders play truthfully. But $o_{0,10m\alpha\eta^2,S,S'}(M) \leq 11\alpha\eta^2$, and the combined value of every other bidder is at most $2\alpha$, so $\MECH$ must allocate all of $M$ to bidder $i$ to satisfy the approximation guarantee, which causes bidder $i^\dagger$ to get no value. Then by NNT, the utility of $i^\dagger$ on this profile under truthful play must be $0$.

    \item Therefore, consider the case where bidders $j \ne i, i^\dagger$ play $\sigma_j(o_{0,0})$, and bidder $i$ plays $\sigma_i(v_{0,0})$ if bidder $i^\dagger$ sends the message $\xi_{i^\dagger}(u)$ in the first round, and $\sigma_i(v_{0,m\alpha\eta^3})$ otherwise. This is a valid strategy because we assumed for contradiction that $\xi_i(v_{0,0}) = \xi_i(v_{0,m\alpha\eta^3})$. Then if bidder $i^\dagger$ has the valuation $o_{0,10m\alpha\eta^2,S,S'}$, they send the message $\xi_{i^\dagger}(o_{0,10m\alpha\eta^2,S,S'}) \ne \xi_{i^\dagger}(u)$ in the first round, so truthful play yields utility $0$. On the other hand, playing $\sigma_{i^\dagger}(u)$ obtains utility at least $o_{0,10m\alpha\eta^2,S,S'}(S) - 2\alpha\eta^2 \geq 8\alpha\eta^2$, contradicting DSIC. \qedhere
\end{itemize}
\end{proof}

We can now finish the proof.~\Cref{lemma:MustSignalSpecialOutcome} states that for all bidders $i \ne i^\dagger$ and valuations $v \in \calV(\calH)$, we have $\xi_i(v_{0,0}) \ne \xi_i(v_{0,m\alpha\eta^3})$. Therefore, construct the chain of valuations
\[
    v_{0,0} \quad v_{0,1} \quad \dots \quad v_{0,m\alpha\eta^3-1} \quad v_{0,m\alpha\eta^3} \enspace ,
\]
where we change from one valuation to the other by incrementing the weight on the special outcome. Let $t \in [m\alpha\eta^3]$ be the first time the message sent changes; then we have the desired result
\[
    \xi_i(v_{0,0}) \quad = \quad \xi_i(v_{0,1}) \quad = \quad \dots \quad = \quad \xi_i(v_{0,t-1}) \quad \ne \quad \xi_i(v_{0,t}) \enspace . \qedhere
\]
\end{proof}

\subsection{Constructing the Simultaneous Algorithm}

\begin{theorem} \label{thm:DSICToSimultaneousAlgorithm}
    Let $\eta \geq 800$ and let $\lambda \in [m^{10}, 2^{m/2}]$. Let $\calH$ be an $(\alpha, \eta, \lambda, \vb{C}, \calU)$-KSSHD, and let $\alpha\eta^2 \leq m$ if $\calH$ uses GS or XOS valuations, and let $\alpha\eta \leq m$ if $\calH$ uses general valuations. Suppose there is a deterministic, IR, NNT, DSIC mechanism $\MECH'$ which attains an $\eta/800$-approximation using $\lambda/m^9$ communication on $\calV[-2^m, 2^m](\calH)$. Then there exists a simultaneous algorithm which attains an $\eta/40$-approximation on $\calH$ using $\lambda/m^9$ communication.
\end{theorem}
\begin{proof}
Let $\MECH$ be the mechanism promised by~\Cref{prop:InformativeNode}. Let $\sigma_1, \dots, \sigma_n$ be the strategy prescriptions of $\MECH$, and let $\xi_1, \dots, \xi_n$ be the first round message prescriptions of $\MECH$.

The randomly sampled inputs are the sets $A_1, \dots, A_\kappa, B \subseteq M$, the random index $\vartheta \in [\lambda]$, the mappings $\varphi_1, \dots, \varphi_\kappa$, the collection $\{A_{k,\ell} : k \in [\kappa], \ell \in [\lambda]\}$, and the bidder input vectors $\vb{x}_1, \dots, \vb{x}_n$.

Condition everything on the public information in $\calH$. If $\calH$ uses GS or XOS valuations, define
\[
    v_{\varphi_k,\vb{x},t}(S) \quad = \quad u_{\vb{x}}(\varphi^{-1}_k(S \cap (A_k \cup B))) + (\alpha\eta + t/m) \cdot \One(\dagger \in S) \enspace ,
\]
and if $\calH$ uses general valuations, define
\[
    v_{\varphi_k,\vb{x},t}(S) \quad = \quad u_{\vb{x}}(\varphi^{-1}_k(S \cap (A_k \cup B))) \cdot \One(S \not\supseteq M) + (2\alpha + t/m) \cdot \One(S \supseteq M) \enspace .
\]

In both cases, define
\[
    \calV_{\varphi_k} \quad = \quad \{v_{\varphi_k,\vb{x},t} : \vb{x} \in \{0,1,2\}^\lambda,\, t \in [0, m^4] \cap \mathbb{Z}\} \quad \subseteq \quad \calV[0](\calH) \enspace .
\]

\begin{definition}[Amplification]
    For $\vb{x} \in \{0,1\}^\lambda$, define the \emph{amplified} random vector $\vb{x}^{\amplify} \in \{0,1,2\}^\lambda$ where independently for all $\ell \in [\lambda]$, we have $x^{\amplify}_\ell \in \{x_\ell,2x_\ell\}$ and $\Pr[x^{\amplify}_\ell = 2 \given x_\ell = 1] = 1/m^3$.
\end{definition}

\begin{definition}[Best-Guarantee Valuation]
    For a bidder $i \in N_k$ and message $\mu_i$, define $\Psi_i(\mu_i)$ to be the set of strategies for bidder $i$ that send $\mu_i$ in the first round of communication. Let $S_i(\vb{\psi})$ be the set allocated to bidder $i$ by $\MECH$ on the strategy profile $\vb{\psi} = (\psi_1, \dots, \psi_n)$. Define the \emph{best-guarantee valuation} $v^*_{i,\vb{\mu}}$ consistent with messages $\vb{\mu} = (\mu_1, \dots, \mu_n)$ by
    \[
        v^*_{i,\vb{\mu}} \quad = \quad \argmax_{v_i \in \calV_{\varphi_k} : \sigma_i(v_i) \in \Psi_i(\mu_i)} \min_{\vb{\psi}_{-i} \in \vb{\Psi_{-i}}(\vb{\mu}_{-i})} v_i(S_i(\sigma_i(v_i), \vb{\psi}_{-i})) \enspace ,
    \]
    and let $\nu^*_{i,\vb{\mu}}$ be the value attained.
\end{definition}

\begin{observation}
    $v^*_{i,\vb{\mu}}$ can be computed in finite time.
\end{observation}

We now describe a simultaneous algorithm for $\calH$ derived from $\MECH$. Recall that we assume the referee knows $A_1, \dots, A_\kappa$.
\begin{enumerate}[topsep=1pt,itemsep=1pt]
    \item Each bidder $i \in N_k$ holds $v_{\varphi_k,\vb{x}_i}$, and independently samples $\vb{x}_i^{\amplify}$. Define $v^{\amplify}_i = v_{\varphi_k,\vb{x}^{\amplify}_i,0}$.
    \item Each bidder $i$ sends $\mu_i = \xi_i(v^{\amplify}_i)$.
    \item The referee computes valuations $v^*_{1,\vb{\mu}}, \dots, v^*_{n,\vb{\mu}}$.
    \item The referee runs $\MECH$ on $\sigma_1(v^*_{1,\vb{\mu}}), \dots, \sigma_n(v^*_{n,\vb{\mu}})$ to obtain the allocation $S^*_1, \dots, S^*_n$.
    \item For all $k \in [\kappa]$, the referee allocates $A_k$ to bidder $i \in N_k$ w.p.\ $v^*_{i,\vb{\mu}}(S^*_i)/\sum_{j \in N_k} v^*_{j,\vb{\mu}}(S^*_j)$. Here, let $0/0 = 0$, in case the denominator is $0$.
\end{enumerate}

It is not clear a priori why this algorithm should work, because $v^*_{i_k,\vb{\mu}}$ need not explicitly reveal that bidder $i_k$ has high value for $A_k$. We will show that typically, the message sent by bidder $i_k$ is consistent with a valuation which has doubled value for $A_k$. Then we can leverage DSIC to convert this consistency into a guarantee that $v^*_{i_k,\vb{\mu}}(S^*_{i_k})$ is large.

\subsection{From Welfare Guarantees to Utility Guarantees: Amplification}

We first show that because there are so many clauses, the messages are unlikely to reveal whether $x^{\amplify}_{i,\vartheta} \ne x_{i,\vartheta}$, conditioned on $x_{i,\vartheta} = 1$. For $k \in [\kappa]$, let $i_k \in N_k$ be the bidder such that $x_{i_k,\vartheta} = 1$, i.e., bidder $i_k$ has high value for $A_k$.

\begin{definition}[Obscured Amplification]
    For $k \in [\kappa]$, let $\calE_k$ be the event that $x^{\amplify}_{i_k,\vartheta} = 1$ and that there exists $\vb{y} \in \{0,1,2\}^\lambda$ such that $y_\vartheta = 2$ and $\xi_{i_k}(v_{\varphi_k,\vb{y},0}) = \xi_{i_k}(v_{\varphi_k,\vb{x}^{\amplify}_{i_k},0})$.
\end{definition}

\begin{lemma} \label{lemma:ObscuredAmplification}
    $\Pr[\calE_k] \geq 1 - 1/m^2$.
\end{lemma}
\begin{proof}
First, note that by definition of $i_k$, we have $x_{i_k,\vartheta} = 1$, so $\Pr[x^{\amplify}_{i_k,\vartheta} = 2] = m^{-3}$.

Denote $\mu_k = \xi_{i_k}(v_{\varphi_k,\vb{x}^{\amplify}_{i_k},0})$. For any first round message $\mu$ that can be sent by bidder $i_k$, let $L_\mu \subseteq [\lambda]$ be the set of indices which $\mu$ precludes from being $2$. In other words, for any $\vb{y} \in \{0,1,2\}^\lambda$ such that $\xi_{i_k}(v_{\varphi_k,\vb{y},0}) = \mu$, we have $y_\ell \ne 2$ for all $\ell \in L_\mu$.

Let $L_k = \{\ell \in [\lambda] : x^{\amplify}_{i_k,\ell} = 1\}$. Conditioned on $x^{\amplify}_{i_k,\vartheta} = 1$, $\vartheta \in L_k$ is uniformly random, so
\begin{align*}
    1 - \Pr[\calE_k] \quad &\leq \quad \Pr[x^{\amplify}_{i_k,\vartheta} = 2] + \Pr[\vartheta \in L_{\mu_k} \given x^{\amplify}_{i_k,\vartheta} = 1] \\
    &\leq \quad m^{-3} + \E[\abs{L_{\mu_k} \cap L_k}/\abs{L_k}] \\
    &\leq \quad m^{-3} + (\lambda/m^5)/(\lambda/m^2) + \Pr[\abs{L_{\mu_k}} \geq \lambda/m^5] + \Pr[\abs{L_k} \leq \lambda/m^2] \enspace .
\end{align*}
Let $\mu$ be a message that rules out at least $\lambda/m^5$ indices. Since $\Pr[y_\ell = 2] \geq 1/(m^3\eta)$ independently for all $\ell \in [\lambda]$,
\[
    \Pr[\mu_k = \mu] \quad \leq \quad (1 - 1/(m^3\eta))^{\lambda/m^5} \quad \leq \quad e^{-\lambda/m^9} \enspace .
\]
In particular, since $\MECH$ runs in $\lambda/m^9$ communication, there are only $2^{\lambda/m^9+1}$ possible messages that bidder $i_k$ can send. Hence, the probability that $\mu_k$ rules out at least $\lambda/m^5$ indices is at most
\[
    2^{\lambda/m^9+1} \cdot e^{-\lambda/m^9} \quad \leq \quad m^{-3} \enspace .
\]
Since $\E[\abs{L_k}] \geq \lambda/(2m)$, a standard Chernoff bound also bounds $\Pr[\abs{L_k} \leq \lambda/m^2] \leq m^{-3}$, so
\[
    1 - \Pr[\calE_k] \quad \leq \quad m^{-3} + m^{-3} + m^{-3} + m^{-3} \quad \leq \quad m^{-2} \enspace . \qedhere
\]
\end{proof}

Let $(S^{\amplify}_1, \dots, S^{\amplify}_n)$ and $(p^{\amplify}_1, \dots, p^{\amplify}_n)$ be the allocation and payment when $\MECH$ is run on the profile $\sigma_1(v^{\amplify}_1), \dots, \sigma_n(v^{\amplify}_n)$. Let $k^{\amplify} \in [\kappa]$ be the choice of $k$ where some $i \in N_k$ has $\dagger \in S^{\amplify}_i$, if such $k$ exists. Also let $k^{\dagger} \in [\kappa]$ be the choice of $k$ where $i^\dagger \in N_k$.

For $k \in [\kappa]$, let $i_k$ be the bidder satisfying $x_{i_k,\vartheta} = 1$. Define
\[
    \Delta^{\GS}_k \quad = \quad \Delta^{\XOS}_k \quad = \quad \abs{S^{\amplify}_{i_k} \cap A_k} - 2\abs{S^{\amplify}_{i_k} \cap B} - 10\alpha/\eta^2 \enspace ,
\]
and
\[
    \Delta^{\GEN}_k \quad = \quad \alpha \cdot \One(S^{\amplify}_{i_k} \supseteq A_k) - 8\abs{S^{\amplify}_{i_k} \cap B} \enspace .
\]

The following lemma shows that when bidder $i_k$ doesn't commit to whether they have doubled value for $A_k$, then DSIC requires $\MECH$ to guarantee them value $\Delta^{\VAL}_k$. The lemma after shows that $\Delta^{\VAL}_k$ is large.

\begin{lemma} \label{lemma:HallucinatedGuarantee}
    Let $k \in [\kappa] \setminus \{k^\dagger, k^{\amplify}\}$, and suppose $\calE_k$ holds. Then when $\calH$ uses valuations of type $\VAL \in \{\GS, \XOS, \GEN\}$, we have $\nu^*_{i_k,\vb{\mu}} \geq \max\{0, \Delta^{\VAL}_k - 1/m\}$.
\end{lemma}
\begin{proof}
By~\Cref{lemma:ObscuredAmplification}, we have $x_{i_k,\vartheta} = 1$, and there exists $\vb{y} \in \{0,1,2\}^\lambda$ such that $y_\vartheta = 2$ and $\xi_{i_k}(v_{\varphi_k,\vb{y},0}) = \xi_{i_k}(v^{\amplify}_{i_k}) \equiv \mu_{i_k}$. In other words, there are two valuations with the same message, and they differ on their coefficient for $A_k$ (and possibly on other sets).

Chaining this with~\Cref{prop:InformativeNode}, there exists $t \in [m^4]$ such that
\[
    \mu_{i_k} \quad = \quad \xi_{i_k}(v_{\varphi_k,\vb{y},0}) \quad = \quad \xi_{i_k}(v_{\varphi_k,\vb{y},t-1}) \quad \ne \quad \xi_{i_k}(v_{\varphi_k,\vb{y},t}) \enspace .
\]
Since $k \ne k^{\amplify}$, we have by definition that $\dagger \not\in S^{\amplify}_{i_k}$, and hence $S^{\amplify}_{i_k}$ does not realize the special outcome (recall for GS and XOS valuations the special outcome is getting $\dagger$, and for general valuations the special outcome is getting $M \ni \dagger$). Hence, the utility $v_{\varphi_k,\vb{y},t}$ gets for playing $\sigma_{i_k}(v^{\amplify}_{i_k})$ is at least
\[
    v_{\varphi_k,\vb{y},t}(S^{\amplify}_{i_k}) - p^{\amplify}_{i_k} \quad \geq \quad v_{\varphi_k,\vb{y},t}(S^{\amplify}_{i_k}) - v^{\amplify}_{i_k}(S^{\amplify}_{i_k}) \quad \geq \quad \Delta^{\VAL}_k \enspace ,
\]
where the first inequality comes from IR, and the second inequality comes from applying Property (4) of~\Cref{def:KSSHD} to the valuations $v_{\varphi_k,\vb{y},t}$ and $v^{\amplify}_{i_k}$, which have doubled and undoubled value for $A_k$, respectively.

Therefore, by DSIC, the utility $v_{\varphi_k,\vb{y},t}$ gets for playing $\sigma_{i_k}(v_{\varphi_k,\vb{y},t})$ is at least $\Delta^{\VAL}_k$. However, this also implies that the utility $v_{\varphi_k,\vb{y},t-1}$ gets for playing $\sigma_{i_k}(v_{\varphi_k,\vb{y},t})$ is at least $\Delta^{\VAL}_k - 1/m$, since $v_{\varphi_k,\vb{y},t}(S) - v_{\varphi_k,\vb{y},t-1}(S) \leq 1/m$ for all $S \subseteq M$.

Finally, suppose there exist strategies $\vb{\psi}_{-i_k} \in \vb{\Psi}_{-i_k}(\vb{\mu}_{-i_k})$ such that $\MECH$ gives less than $\Delta^{\VAL}_k - 1/m$ value to bidder $i_k$ with valuation $v_{\varphi_k,\vb{y},t-1}$ on the profile $\sigma_{i_k}(v_{\varphi_k,\vb{y},t-1}), \vb{\psi}_{-i_k}$. Then define the strategy $\psi'_i$ by ``play $\psi_i$ if bidder $i_k$ sends the message $\mu_{i_k}$ in the first round, and play $\sigma_i(v^{\amplify}_i)$ otherwise.'' This is well-defined because $\psi_i$ and $\sigma_i(v^{\amplify}_i)$ send the same message in the first round.

Then because $\mu_{i_k} = \xi_{i_k}(v_{\varphi_k,\vb{y},t-1}) \ne \xi_{i_k}(v_{\varphi_k,\vb{y},t})$, bidder $i_k$ playing $\sigma_{i_k}(v_{\varphi_k,\vb{y},t})$ when they actually have the valuation $v_{\varphi_k,\vb{y},t-1}$ causes them to get strictly more utility against the strategies $\vb{\psi}'_{-i_k}$, which contradicts DSIC.

Thus, the valuation $v_{\varphi_k,\vb{y},t-1}$ satisfies $\sigma_{i_k}(v_{\varphi_k,\vb{y},t-1}) \in \Psi_{i_k}(\mu_{i_k})$, and the strategy $\sigma_{i_k}(v_{\varphi_k,\vb{y},t-1})$ guarantees value at least $\max\{0, \Delta^{\VAL}_k - 1/m\}$, since values are non-negative.
\end{proof}

\begin{lemma} \label{lemma:HallucinatedWelfare}
    \[
        \sum_{k \in [\kappa]} \Delta^{\VAL}_k \quad \geq \quad 180m/\eta \enspace .
    \]
\end{lemma}
\begin{proof}
Reference will be made to the properties of~\Cref{def:KSSHD}.

When $\VAL \in \{\GS, \XOS\}$, by the fact that $\dagger$ contributes at most $\alpha\eta$ value and by~\Cref{lemma:MustAllocateToCorrectBiddersXOS},
\begin{align*}
    \sum_{i \in [n]} v^{\amplify}_i(S^{\amplify}_i) \quad &\leq \quad \alpha\eta + 17m/\eta + 2\sum_{k \in [\kappa]} \abs{S^{\amplify}_{i_k} \cap A_k} \\
    &\leq \quad \alpha\eta + 17m/\eta + 4\beta + 20\kappa\alpha/\eta^2 + 2\sum_{k \in [\kappa]} \Delta^{\VAL}_k \\
    &\leq \quad 40m/\eta + 2\sum_{k \in [\kappa]} \Delta^{\VAL}_k \enspace .
\end{align*}

When $\VAL = \GEN$, the special outcome is exclusive and only obtains $2\alpha \leq 2m/\eta$ value. Therefore, by~\Cref{lemma:MustAllocateToCorrectBiddersGEN},
\begin{align*}
    \sum_{i \in [n]} v^{\amplify}_i(S^{\amplify}_i) \quad &\leq \quad 2m/\eta + 2\alpha\sum_{k \in [\kappa]} \One(S^{\amplify}_{i_k} \supseteq A_k) \\
    &\leq \quad 2m/\eta + 16\beta + 2\sum_{k \in [\kappa]} \Delta^{\VAL}_k \\
    &\leq \quad 40m/\eta + 2\sum_{k \in [\kappa]} \Delta^{\VAL}_k \enspace .
\end{align*}

Since the optimal welfare is at least $\kappa\alpha = m/2$, the welfare obtained by $\MECH$ must be at least $400m/\eta$. Substituting and simplifying yields the desired result.
\end{proof}

We can now finally finish the proof. By~\Cref{lemma:ObscuredAmplification}, the events $\calE_1, \dots, \calE_\kappa$ jointly hold w.p.\ at least $1-\kappa/m^2 \geq 1-1/m$. Condition on this occurring.

Let $k^*$ be the choice of $k$ such that some $i \in N_k$ satisfies $\dagger \in S^*_i$, if such $k$ exists, and let $K = [\kappa] \setminus\{k^\dagger, k^{\amplify}, k^*\}$. Let $T^*_k = \bigcup_{i \in N_k} S^*_i$. Then the simultaneous algorithm attains welfare at least
\begin{align*}
    \sum_{k \in K} \frac{v^*_{i_k,\vb{\mu}}(S^*_{i_k})}{\sum_{j \in N_k} v^*_{j,\vb{\mu}}(S^*_j)} \alpha \quad &\geq \quad \sum_{k \in K} \frac{\alpha \cdot v^*_{i_k,\vb{\mu}}(S^*_{i_k})}{2 \max\{1, \abs{T^*_k \cap (A_k \cup B)}\}} \\
    &\geq \quad \sum_{k \in K} \frac{\max\{0, \Delta^{\VAL}_k\}}{2} \frac{\alpha}{\max\{1, \abs{T^*_k \cap (A_k \cup B)}\}} - \frac{\kappa\alpha}{m} \\
    &\geq \quad \sum_{k \in K} \frac{\max\{0, \Delta^{\VAL}_k\}}{2} \frac{\alpha}{\alpha + \abs{T^*_k \cap B}} - \frac{1}{2} \\
    &\geq \quad \sum_{k \in K} \frac{\max\{0, \Delta^{\VAL}_k\}}{2} \bigg(1 - \frac{\abs{T^*_k \cap B}}{\alpha}\bigg) - \frac{1}{2} \\
    &\geq \quad \frac{1}{2} \sum_{k \in K} \Big(\Delta^{\VAL}_k - \abs{T^*_k \cap B}\Big) - \frac{1}{2} \\
    &\geq \quad \frac{1}{2} \bigg(\sum_{k \in K} \Delta^{\VAL}_k - \abs{B} - 1\bigg) \\
    &\geq \quad \frac{1}{2} \bigg(\sum_{k \in [\kappa]} \Delta^{\VAL}_k - \beta - 3\alpha - 1\bigg) \\
    &\geq \quad 85m/\eta \enspace ,
\end{align*}
where the first line comes from the fact that the average value of an item is at most $2$, the second line comes from~\Cref{lemma:HallucinatedGuarantee}, the fourth line comes from the inequality $1/(1+x) \geq 1-x$, the fifth line comes from $\Delta^{\VAL}_k \leq \alpha$, and the last line comes from~\Cref{lemma:HallucinatedWelfare}. After de-conditioning, the welfare guarantee is still at least $80m/\eta$.

We conclude by observing that the optimal welfare of the KSSHD is at most $2m$, and hence the simultaneous algorithm yields the desired $\eta/40$-approximation.
\end{proof}

\subsection{Wrapping Up}

Applying each of Propositions~\ref{prop:GSKSSHD},~\ref{prop:XOSKSSHD},~\ref{prop:GeneralKSSHD} with~\Cref{thm:KSSHDIsHard,thm:DSICToSimultaneousAlgorithm} gives us that for $\gamma \in [\log m, m/10^{100}]$, we have
\begin{itemize}[topsep=1pt,itemsep=1pt]
    \item $\DSICgs(m, \gamma) = \Omega(((m-1)/(\gamma + 9\log m))^{1/7}) = \Omega((m/\gamma)^{1/7})$,
    \item $\DSICxos(m, \gamma) = \Omega(((m-1)/(\gamma + 9\log m))^{1/5}) = \Omega((m/\gamma)^{1/5})$,
    \item $\DSICgen(m, \gamma) = \Omega((m-1)/(\gamma + 9\log m)) = \Omega(m/\gamma)$.
\end{itemize}
For $\gamma \in [m/10^{100}, m]$, we have $\DSICgs(m, \gamma), \DSICxos(m, \gamma), \DSICgen(m, \gamma) \geq 1 = \Omega(m/\gamma)$ trivially. This completes the proofs of Theorems~\ref{thm:GSDSICHardness},~\ref{thm:XOSDSICHardness},~\ref{thm:GeneralDSICHardness}.

\section{Conclusion}

We prove polynomial hardness-of-approximation results for deterministic DSIC mechanisms on GS, XOS, and general valuations. In the appendix, we prove for the two-bidder setting that every universally DSIC mechanism attaining a $0.9999$-approximation on weighted matroid rank valuations must use exponential communication, while also giving a polynomial-communication universally DSIC mechanism that beats a $1/2$-approximation for submodular valuations. Additionally, we show how to beat a $3/4$-approximation with a polynomial-communication simultaneous algorithm on matroid rank valuations. There are a few open questions.

\vspace*{-1em}
\paragraph{Two-Bidder Settings.}
Although we beat a $1/2$-approximation for submodular valuations with a universally DSIC mechanism, it is still open whether a polynomial-communication truthful mechanism can beat a $1/2$-approximation for XOS valuations. Additionally, it remains open to design a polynomial-communication deterministic DSIC mechanism beating a $1/2$-approximation for weighted matroid rank valuations, or to design a polynomial-communication deterministic truthful mechanism beating a $1/2$-approximation for submodular valuations.

\vspace*{-1em}
\paragraph{Improved Simultaneous Hard Instances.}
Our DSIC mechanism to simultaneous algorithm reduction allows us to recover an $\Omega((m/\gamma)^r)$ DSIC hardness-of-approximation whenever we have an $\Omega((m/\gamma)^r)$ KSSHD (\Cref{def:KSSHD}) hardness-of-approximation for a valuation class richer than GS. Therefore, a natural direction is designing better KSSHDs, or KSSHDs for different valuation classes. At the very least, there is a simultaneous $O(m^{1/3})$-approximation algorithm and a simultaneous $\Omega(m^{1/4-\varepsilon})$ hardness-of-approximation for subadditive valuations~\cite{DobzinskiNO14}, so it is plausible that the DSIC hardness-of-approximation bounds can be improved. Note that the hard instance of~\cite{DobzinskiNO14} does not conform to the KSSHD structure and so cannot be used directly.~\cite{Assadi17} gives a simultaneous $\Omega(m^{1/3-\varepsilon})$ hardness-of-\emph{estimation} for XOS valuations, meaning that it is hard to determine the welfare within a factor of $\Omega(m^{1/3-\varepsilon})$. The natural modification for hardness-of-approximation is stuck at $\Omega(m^{1/5-\varepsilon})$.

\vspace*{-1em}
\paragraph{Beyond Reductions to Simultaneous Lower Bounds.}
As previously stated, we prove that for general valuations, the optimal approximation guarantee for deterministic DSIC mechanisms given $2^\gamma$ communication is $\Theta(m/\gamma)$. This happens to coincide with the optimal approximation guarantee for deterministic truthful mechanisms which are either MIR, implementable with value queries, or computationally efficient~\cite{QiuW24,BlumrosenN05a,DanielySS15}.

For submodular, XOS, and subadditive valuations, the optimal approximation guarantee for deterministic truthful mechanisms which are either MIR, implementable with value queries, or computationally efficient is $\tilde{\Theta}((m/\gamma)^{1/2})$~\cite{QiuW24,DobzinskiV16}. However, simultaneous algorithms for subadditive valuations can attain an $O(m^{1/3})$-approximation. Thus, if the correct answer is also $\tilde{\Theta}((m/\gamma)^{1/2})$ for deterministic DSIC mechanisms, new techniques would be required for a tight lower bound.

\section{Acknowledgements}

We are very grateful to Jan Vondr\'{a}k for a number of helpful discussions over the course of the project and for comments on an earlier draft of the paper. We also thank Shahar Dobzinski for suggesting that the two-bidder simultaneous lower bound could be used for a DSIC lower bound.

\section{AI Usage Statement}

The results in this paper were proved with assistance from GPT 5.6 Pro. The two-bidder DSIC lower bound was largely generated by GPT after being given the two-bidder simultaneous lower bound, with the author supplying ideas to generalize to the many-bidder DSIC lower bound. The idea to use the EP2A for weighted matroid rank valuations was proposed by GPT, with further interaction generalizing the mechanism to submodular valuations. The paper and proof presentation were drafted by hand, with GPT being used to check for bugs and typos. The author takes full responsibility for the content of the paper.

\bibliographystyle{alpha}
\bibliography{MasterBib}

\appendix

\section{Simultaneous Hardness for Two Matroid Rank Valuations} \label{sec:TwoBidderSimultaneousHardness}

In this section, we show that even for two matroid rank valuations, no arbitrarily good approximation is possible for simultaneous algorithms. The construction in this section also serves as a warm-up for the DSIC hard instance for two weighted matroid rank functions.

We will reduce from the INDEX problem, where Alice holds a bit string $x \in \{0, 1\}^\lambda$ and Bob holds an index $k \in [\lambda]$. The goal is for Alice to send a message to Bob so that Bob can recover $x_k$. Any protocol that succeeds on INDEX w.p.\ $1/2 + \Omega(1)$ requires $\Omega(\lambda)$ communication~\cite{Ablayev96}.

The idea is to have a list of pairs of sets $(S_{1,0}, S_{1,1}), \dots, (S_{\lambda,0}, S_{\lambda,1})$ such that for every $\vb{x} \in \{0, 1\}^\lambda$, it is possible to construct a matroid for Alice such that the rank of the sets $S_{\ell,x_\ell}$ is low for all $\ell \in [\lambda]$. Then Bob will hold a matroid which, for some $k \in [\lambda]$, has high value for both $S_{k,0}, S_{k,1}$, and hence a good approximation requires correctly giving Alice the set in $S_{k,0}, S_{k,1}$ which she has high value for. This requires determining Alice's bit $x_k$ w.p.\ bounded away from $1/2$.

\subsection{A Hard Matroid}

Fix any constant $\varepsilon \in (0, 1/16]$, and let $\lambda = 2^{\varepsilon^2 m}$.

\begin{lemma} \label{lemma:IndexSetsExist}
    There exist partitions $\{(A_\ell, B_\ell, C_{\ell,0}, C_{\ell,1}) : \ell \in [\lambda]\}$ such that
    \begin{enumerate}[topsep=1pt,itemsep=1pt]
        \item $\abs{A_\ell} = \abs{B_\ell} = (3/8 + \varepsilon)m$ and $\abs{C_{\ell,0}} = \abs{C_{\ell,1}} = (1/8 - \varepsilon)m$ for all $\ell \in [\lambda]$.
        \item $\abs{(A_\ell \cup C_{\ell,b}) \cap (A_{\ell'} \cup C_{\ell',b'})} \leq (1/4 + 2\varepsilon)m$ for all $b, b' \in \{0, 1\}$ and distinct $\ell, \ell' \in [\lambda]$.
    \end{enumerate}
\end{lemma}
\begin{proof}
Take independently random partitions satisfying (1). Since $A_\ell \cup C_{\ell,b}$ and $A_{\ell'} \cup C_{\ell',b'}$ are independent for $\ell \ne \ell'$, their expected intersection is $m/4$, and a Chernoff bound and union bound yield the desired list of partitions.
\end{proof}

Fix $\{(A_\ell, B_\ell, C_{\ell,0}, C_{\ell,1}) : \ell \in [\lambda]\}$ satisfying the properties of~\Cref{lemma:IndexSetsExist}.

\begin{definition}[Alice-Type]
    For $\vb{x} \in \{0, 1\}^\lambda$, define
    \[
        \calA_{\vb{x}} \quad = \quad \{S \subseteq M : \abs{S} \leq m/2,\, \abs{S \cap (A_\ell \cup C_{\ell,x_\ell})} \leq (3/8 + \varepsilon)m \; \forall \ell \in [\lambda]\} \enspace .
    \]
\end{definition}

\begin{definition}[Bob-Type]
    For $k \in [\lambda]$, define
    \[
        \calB_k \quad = \quad \{S \subseteq M : S \cap A_k = \emptyset,\, \abs{S \cap (C_{k,0} \cup C_{k,1})} \leq (1/8 - \varepsilon) m\} \enspace .
    \]
\end{definition}

Since Bob-types are just the direct sum of an empty matroid on $A_k$, the free matroid on $B_k$, and a uniform matroid on $C_{k,0} \cup C_{k,1}$, they are matroids. We now prove that Alice-types are matroids.

\begin{proposition}
    $\calA_{\vb{x}}$ is a matroid for all $\vb{x} \in \{0, 1\}^\lambda$.
\end{proposition}
\begin{proof}
Downward-closure is immediate, so it only remains to check augmentation.

Let $S, T \in \calA_{\vb{x}}$ and $\abs{S} < \abs{T}$. We say that constraint $\ell \in [\lambda]$ is \emph{tight} for $S$ if $\abs{S \cap (A_\ell \cup C_{\ell,x_\ell})} = (3/8 + \varepsilon)m$. If two distinct constraints $\ell, \ell' \in [\lambda]$ were tight for $S$, then by condition (2) of~\Cref{lemma:IndexSetsExist},
\begin{align*}
    \abs{S} \quad &\geq \quad \abs{S \cap (A_\ell \cup C_{\ell,x_\ell})} + \abs{S \cap (A_{\ell'} \cup C_{\ell',x_{\ell'}})} - \abs{(A_\ell \cup C_{\ell,x_\ell}) \cap (A_{\ell'} \cup C_{\ell',x_{\ell'}})} \\
    &= \quad (3/8 + \varepsilon)m + (3/8 + \varepsilon)m - (1/4 + 2\varepsilon)m \\
    &\geq \quad \abs{T} \enspace ,
\end{align*}
so at most one constraint can be tight for $S$.

Suppose for contradiction that no element of $T \setminus S$ augments $S$. Then since $\abs{S} < m/2$, it must be the case that some constraint $\ell \in [\lambda]$ is tight for $S$, and $T \setminus S \subseteq A_\ell \cup C_{\ell,x_\ell}$. But this means $T$ must have more items in $A_\ell \cup C_{\ell,x_\ell}$ than $S$, which is impossible because constraint $\ell$ is tight for $S$.
\end{proof}

\subsection{Reduction from INDEX}

For $\vb{x} \in \{0, 1\}^\lambda$, let $u_{\vb{x}} : 2^M \to \mathbb{R}_+$ be the rank function of $\calA_{\vb{x}}$, and for $k \in [\lambda]$, let $v_k : 2^M \to \mathbb{R}_+$ be the rank function of $\calB_k$.

\begin{lemma} \label{lemma:TwoBidderWelfareSimultaneous}
    The allocation which gives $S$ to $u_{\vb{x}}$ and $\closure{S}$ to $v_k$ has welfare at most
    \[
        m - (\abs{S \cap C_{k,x_k}} + \abs{\closure{S} \cap C_{k,1-x_k}})/2 \enspace .
    \]
\end{lemma}
\begin{proof}
The items in $B_k$ have marginal value $1$ for $v_k$, and the items in $A_k$ have marginal value $0$ for $v_k$. Hence, the welfare does not decrease by assuming $S \supseteq A_k$ and $\closure{S} \supseteq B_k = \emptyset$. Then observe that
\[
    u_{\vb{x}}(S) \quad = \quad (3/8 + \varepsilon) m + \abs{S \cap C_{k,1-x_k}} \quad = \quad m/2 - \abs{\closure{S} \cap C_{k,1-x_k}}
\]
and
\begin{align*}
    v_k(\closure{S}) \quad &= \quad (3/8 + \varepsilon) m + \min\{(1/8 - \varepsilon) m, \abs{\closure{S} \cap C_{k,1-x_k}} + \abs{\closure{S} \cap C_{k,x_k}}\} \\
    &= \quad m/2 + \min\{0, \abs{\closure{S} \cap C_{k,1-x_k}} - \abs{S \cap C_{k,x_k}}\} \enspace .
\end{align*}
Therefore,
\begin{align*}
    u_{\vb{x}}(S) + v_k(\closure{S}) \quad &= \quad m + \min\{-\abs{\closure{S} \cap C_{k,1-x_k}}, -\abs{S \cap C_{k,x_k}}\} \\
    &\leq \quad m - (\abs{S \cap C_{k,x_k}} + \abs{\closure{S} \cap C_{k,1-x_k}})/2 \enspace . \qedhere
\end{align*}
\end{proof}

We now reduce from INDEX. Let Alice hold $\vb{x} \in \{0, 1\}^\lambda$, and let Bob hold $k \in [\lambda]$. Let $\ALG$ be a simultaneous algorithm attaining a $(15/16+\varepsilon)$-approximation for two matroid rank valuations. Alice sends Bob the message she would send in $\ALG$ with $u_{\vb{x}}$.

With Alice's message, Bob can compute the allocation $S, \closure{S}$ made by $\ALG$. Bob then samples a uniformly random element $j \in C_{k,0} \cup C_{k,1}$, and outputs $\One(j \in (S \cap C_{k,0}) \cup (\closure{S} \cap C_{k,1}))$.

Observe that Bob errs precisely when $j \in (S \cap C_{k,x_k}) \cup (\closure{S} \cap C_{k,1-x_k})$. Therefore, by~\Cref{lemma:TwoBidderWelfareSimultaneous} and the approximation guarantee of $\ALG$, Bob errs w.p.\ at most
\[
    \frac{\E[\abs{(S \cap C_{k,x_k}) \cup (\closure{S} \cap C_{k,1-x_k})}]}{\abs{C_{k,0} \cup C_{k,1}}} \quad \leq \quad \frac{(1/8 - 2\varepsilon)m}{2(1/8 - \varepsilon)m} \quad \leq \quad \frac{1}{2} - 2\varepsilon \enspace .
\]
As $\varepsilon > 0$ is a constant, this protocol solves INDEX w.p.\ $1/2 + \Omega(1)$. Thus, $\ALG$ uses $\Omega(\lambda) = 2^{\Theta(m)}$ communication, which proves~\Cref{thm:TwoBidderSimultaneousHardness}.

\section{DSIC Hardness for Two Weighted Matroid Rank Valuations}

We will now leverage the construction of~\Cref{sec:TwoBidderSimultaneousHardness} to prove an impossibility result for DSIC mechanisms on two weighted matroid rank valuations. The high-level ideas are very similar to those in the many-bidder proof. To handle randomization, we will define a hard distribution, then reason about the deterministic DSIC mechanism in the support of the universally DSIC mechanism which performs best on the hard distribution. We will make use of the following result.

\begin{fact}[\cite{DobzinskiRV22}] \label{fact:MinimalMechanism}
    For every deterministic DSIC mechanism $\MECH'$ and restricted domain $\calV_1 \times \dots \times \calV_n$, there exists a deterministic DSIC mechanism $\MECH$ using the same communication and realizing the same allocation and payment rule on $\calV_1 \times \dots \times \calV_n$, such that for every non-leaf node $\calN$ in the communication protocol of $\MECH$, there exists a bidder $i$ and $v, v' \in \calV_i$ where $\MECH$ prescribes different messages for $v, v'$ at $\calN$.
\end{fact}

\subsection{The Hard Distribution}

Add two special items $\dagger_1, \dagger_2$, so that the total item set is $M = [m] \cup \{\dagger_1, \dagger_2\}$.

Let $\{(A_\ell, B_\ell, C_{\ell,0}, C_{\ell,1}) : \ell \in [\lambda]\}$ be partitions of $[m]$ satisfying the properties of~\Cref{lemma:IndexSetsExist} for $\varepsilon = 1/40$, so $3/8 + \varepsilon = 4/10$, $1/8 - \varepsilon = 1/10$, and $1/4 + 2\varepsilon = 3/10$. For all $\vb{x} \in \{0, 1\}^\lambda$, let $r_{\vb{x}} : 2^{[m]} \to \mathbb{R}_+$ be the Alice-type matroid rank valuation.

For $i \in \{1, 2\}$, we define the following weighted matroid rank valuations:
\begin{itemize}[topsep=1pt,itemsep=1pt]
    \item For $\vb{x} \in \{0, 1\}^\lambda$, define $u_{i,\vb{x}}$ by
    \[
        u_{i,\vb{x}}(S) \quad = \quad r_{\vb{x}}(S \cap [m]) + 10m \cdot \One(\dagger_i \in S) \enspace .
    \]
    \item For $k \in [\lambda]$ and $t \in \{0, 1\}$, define $v_{i,k,t}$ by
    \[
        v_{i,k,t}(S) \quad = \quad 2\abs{S \cap B_k} + (1 + t)\abs{S \cap C_{k,0}}/3 + 10m \cdot \One(\dagger_i \in S) + 5mt \cdot \One(\dagger_{-i} \in S) \enspace .
    \]
    \item For $t \in \{0, 1\}$, define $w_{i,t}$ by
    \[
        w_{i,t}(S) \quad = \quad m(1 + 3t) \cdot \One(\dagger_i \in S) + 9m(1 - t) \cdot \One(\dagger_{-i} \in S) \enspace .
    \]
\end{itemize}

Define
\[
    \calF_i \quad = \quad \{u_{i,\vb{x}} : \vb{x} \in \{0, 1\}^\lambda\} \cup \{v_{i,k,t} : k \in [\lambda],\, t \in \{0, 1\}\} \enspace .
\]

For $i \in \{1, 2\}$, let $\calD_{i,u}$ be the following distribution over instances: sample random $\vb{x} \in \{0, 1\}^\lambda$ and $k \in [\lambda]$, and let bidder $i$ have $u_{i,\vb{x}}$ and bidder $-i$ have $v_{-i,k,x_k}$.

For $i \in \{1, 2\}$, let $\calD_{i,w}$ be the following distribution over instances: sample random $k \in [\lambda]$ and $t \in \{0, 1\}$, and let bidder $i$ have $w_{i,t}$ and bidder $-i$ have $v_{-i,k,t}$.

The hard distribution $\calD$ will mix the distributions $\calD_{1,u}, \calD_{2,u}, \calD_{1,w}, \calD_{2,w}$ uniformly.

The proof sketch will be as follows.
\begin{enumerate}[topsep=1pt,itemsep=1pt]
    \item Doing well on $\calD_{i,w}$ implies that bidder $-i$ greatly influences the price of $\dagger_i$. This is all we need from $\calD_{i,w}$, and the rest of the argument works on $\calD_{i,u}$ (so bidder $i$ has the Alice-type valuation and bidder $-i$ has the additive valuation).
    \item Suppose that bidder $i$ sends two different messages in the first round. Then bidder $-i$ can play a contingent strategy that raises the price of $\dagger_i$ only when bidder $i$ sends a specific message. Since the value bidder $i$ has for $\dagger_i$ is too high, the only way for the mechanism to be DSIC is if bidder $-i$ is forced to commit to their value for $\dagger_i$.
    \item Now that bidder $-i$ must send information in the first round, bidder $i$ can play a contingent strategy against bidder $-i$. In particular, bidder $i$ can play a contingent strategy if two Alice-type valuations send the same message while having opposite values of $x_k$. DSIC then enforces a utility guarantee for bidder $-i$, regardless of the value of $x_k$. But this results in a large loss of welfare because the optimal allocation depends on $x_k$.
    \item Thus, the first message of bidder $i$ usually reveals the value of $x_k$, which gives a simultaneous algorithm to solve INDEX w.p.\ $1/2 + \Omega(1)$. This proves the communication lower bound.
\end{enumerate}

\subsection{Properties of the Hard Distribution}

We first prove a lower bound on the welfare loss on the instances from $\calD_{i,u}$. We mainly use this result in Step (3) of the proof sketch, but we also need it to use~\Cref{fact:MinimalMechanism}, so we prove it now.

For $i \in \{1, 2\}$, $\vb{x} \in \{0, 1\}^\lambda$, and $k \in [\lambda]$, let $\Delta_{i,\vb{x},k}(S_i, S_{-i})$ denote the difference between the optimal welfare and the welfare of the allocation $S_i, S_{-i}$ on the instance $u_{i,\vb{x}}, v_{-i,k,x_k}$.

\begin{lemma} \label{lemma:TwoBidderWelfareGuarantee}
    For all $i \in \{1, 2\}$, $\vb{x} \in \{0, 1\}^\lambda$, and $k \in [\lambda]$,
    \[
        \Delta_{i,\vb{x},k}(S_i, S_{-i}) \quad \geq \quad \begin{cases}
            \abs{C_{k,0} \setminus S_{-i}}/3 & x_k = 0 \\
            \abs{C_{k,0} \cap S_{-i}}/3 & x_k = 1
        \end{cases} \enspace .
    \]
\end{lemma}
\begin{proof}
There is only one way to optimally allocate $\dagger_1, \dagger_2$. Aside from the special items, bidder $-i$ only values the items in $B_k \cup C_{k,0}$. Therefore, an optimal allocation can give $A_k \cup C_{k,1}$ to bidder $i$.

When $x_k = 0$, bidder $i$ receives maximum non-special value $m/2$ from $A_k \cup C_{k,1}$, and hence receives marginal value $0$ from every item in $C_{k,0}$. Therefore, since the marginal value of bidder $-i$ for items in $C_{k,0}$ is $1/3$, we have $\Delta_{i,\vb{x},k}(S_i, S_{-i}) \geq \abs{C_{k,0} \setminus S_{-i}}/3$.

When $x_k = 1$, bidder $i$ receives only value $4m/10$ from $A_k \cup C_{k,1}$, and hence receives marginal value $1$ from every item in $C_{k,0}$. Therefore, since the marginal value of bidder $-i$ for items in $C_{k,0}$ is only $2/3$, we have $\Delta_{i,\vb{x},k}(S_i, S_{-i}) \geq \abs{C_{k,0} \cap S_{-i}}/3$.
\end{proof}

\begin{corollary} \label{cor:TwoBidderWelfareGuarantee}
    Let $i \in \{1, 2\}$, $\vb{x}, \vb{y} \in \{0, 1\}^\lambda$, and $k \in [\lambda]$. If $x_k = 0$ and $y_k = 1$, then for allocations $S_i, S_{-i}$ and $T_i, T_{-i}$ such that $\abs{C_{k,0} \cap S_{-i}} \leq \abs{C_{k,0} \cap T_{-i}}$,
    \[
        \Delta_{i,\vb{x},k}(S_i, S_{-i}) + \Delta_{i,\vb{y},k}(T_i, T_{-i}) \quad \geq \quad m/30 \enspace .
    \]
\end{corollary}
\begin{proof}
By~\Cref{lemma:TwoBidderWelfareGuarantee},
\begin{align*}
    \Delta_{i,\vb{x},k}(S_i, S_{-i}) + \Delta_{i,\vb{y},k}(T_i, T_{-i}) \quad &\geq \quad \abs{C_{k,0} \setminus S_{-i}}/3 + \abs{C_{k,0} \cap T_{-i}}/3 \\
    &= \quad \abs{C_{k,0}}/3 - \abs{C_{k,0} \cap S_{-i}}/3 + \abs{C_{k,0} \cap T_{-i}}/3 \\
    &\geq \quad m/30 \enspace . \qedhere
\end{align*}
\end{proof}

We now prove that when a mechanism obtains good welfare on the distribution $\calD_{i,w}$, the price paid by bidder $i$ for the item $\dagger_i$ is heavily influenced by the value that bidder $-i$ has for $\dagger_i$. This is the bulk of the argument in Step (1) of the proof sketch.

\begin{definition}
    Let $i \in \{1, 2\}$ and $k \in [\lambda]$, and let $\MECH'$ be a deterministic DSIC mechanism. We say that $\MECH'$ is \emph{safe} for $(i, k)$ if
    \begin{itemize}[topsep=1pt,itemsep=1pt]
        \item $\MECH'$ allocates $\dagger_i$ to bidder $i$ and $\dagger_{-i}$ to bidder $-i$ on the instance $w_{i,0}, v_{-i,k,0}$.
        \item $\MECH'$ allocates $\dagger_i, \dagger_{-i}$ to bidder $-i$ on the instance $w_{i,1}, v_{-i,k,1}$.
    \end{itemize}
\end{definition}

\begin{lemma} \label{lemma:SafePaymentLow}
    Let $\MECH'$ be a deterministic, IR, NNT, DSIC mechanism which is safe for $(i, k) \in \{1, 2\} \times [\lambda]$. Then on all instances $f_i, v_{-i,k,0}$ where $f_i \in \calF_i$, bidder $i$ pays at most $2m$ and is allocated $\dagger_i$ but not $\dagger_{-i}$.
\end{lemma}
\begin{proof}
Let $f_i \in \calF_i$, and observe that $f_i(\{\dagger_i\}) = 10m$, $f_i(\{\dagger_{-i}\}) \leq 5m$, and $f_i([m]) \leq m$. Let $S_f, p_f$ be the allocation and payment of bidder $i$ on instance $f_i, v_{-i,k,0}$, and let $S_w, p_w$ be the allocation and payment of bidder $i$ on instance $w_{i,0}, v_{-i,k,0}$.

By safety of $(i, k)$ and IR, we have $p_w \leq w_{i,0}(S_w) \leq m$. Therefore, if $\dagger_i \not\in S_f$, we have $f_i(S_f) \leq f_i([m] \cup \{\dagger_{-i}\}) \leq 6m$. By NNT, the utility of bidder $i$ with $f_i$ playing truthfully is at most $6m$, while the utility of playing the strategy of $w_{i,0}$ is at least $f_i(\{\dagger_i\}) - p_w \geq 9m$. This contradicts $\MECH'$ being DSIC, and hence $\dagger_i \in S_f$.

Now, suppose for contradiction that $\dagger_{-i} \in S_f$, along with $\dagger_i$. Then by DSIC,
\[
    w_{i,0}(S_w) - p_w \quad \geq \quad w_{i,0}(S_f) - p_f \qquad\quad \implies \quad\qquad p_f \quad \geq \quad w_{i,0}(\{\dagger_{-i}\}) \quad = \quad 9m \enspace .
\]
Therefore, the utility of bidder $i$ with $f_i$ under truthful play is at most $f_i(M) - p_f \leq 7m$, while the utility of playing the strategy of $w_{i,0}$ is at least $9m$, which again contradicts $\MECH'$ being DSIC. Hence, $\dagger_i \in S_f$ and $\dagger_{-i} \not\in S_f$.

Finally, DSIC and safety of $(i, k)$ give
\[
    f_i(S_f) - p_f \quad \geq \quad f_i(S_w) - p_w \qquad\quad \implies \quad\qquad p_f \quad \leq \quad f_i(S_f) - f_i(S_w) + p_w \quad \leq \quad 2m \enspace . \qedhere
\]
\end{proof}

\begin{lemma} \label{lemma:SafePaymentHigh}
    Let $\MECH'$ be a deterministic, IR, NNT, DSIC mechanism which is safe for $(i, k) \in \{1, 2\} \times [\lambda]$. Then on all instances $f_i, v_{-i,k,1}$ where $f_i \in \calF_i$, bidder $i$ pays at least $4m$ when $\MECH'$ allocates $\dagger_i$ to bidder $i$.
\end{lemma}
\begin{proof}
Suppose for contradiction that there exists $f_i \in \calF_i$ such that bidder $i$ receives $\dagger_i$ and pays less than $4m$ on the instance $f_i, v_{-i,k,1}$. Then since $w_{i,1}(\{\dagger_i\}) = 4m$, if bidder $i$ with $w_{i,1}$ plays the strategy of $f_i$, they obtain positive utility, whereas truthful play yields no utility by safety of $(i, k)$ and NNT. This contradicts $\MECH'$ being DSIC.
\end{proof}

\subsection{How DSIC Mechanisms Behave on the Hard Distribution}

Let $\MECH''$ be an IR, NNT, universally DSIC mechanism which attains a $0.9999$-approximation for two weighted matroid rank valuations. Since $\calD$ is supported on instances of two weighted matroid rank valuations, $\MECH''$ attains a $0.9999$-approximation in expectation over $\calD$.

Let $\MECH'$ be the deterministic, IR, NNT, DSIC mechanism in the support of $\MECH''$ which attains the best approximation in expectation over $\calD$. In particular, $\MECH'$ attains a $0.9999$-approximation in expectation on $\calD$, and hence at least a $0.9996$-approximation in expectation over each of $\calD_{1,u}, \calD_{2,u}, \calD_{1,w}, \calD_{2,w}$.

Applying~\Cref{fact:MinimalMechanism}, let $\MECH$ be the deterministic, IR, NNT, DSIC mechanism which behaves identically to $\MECH'$ on the restricted domain of $\calF_1 \times \calF_2$, and where every non-leaf node in the communication protocol has some bidder $i$ being prescribed at least two different messages. $\MECH$ attains a $0.9996$-approximation in expectation over both $\calD_{1,u}, \calD_{2,u}$, because these distributions are supported on $\calF_1 \times \calF_2$.

\begin{lemma} \label{lemma:RootIsNotLeaf}
    The root node of the communication protocol of $\MECH$ is not a leaf.
\end{lemma}
\begin{proof}
Assume for contradiction that the root node is a leaf. Then $\MECH$ performs a constant allocation $S_i, S_{-i}$ under truthful play.

Fix $i \in \{1, 2\}$, and consider the outcomes $(\vb{x}, k) \in \{0, 1\}^\lambda \times [\lambda]$ in $\calD_{i,u}$. Consider the coupling which pairs $(\vb{x}, k)$ and $(\vb{y}, k)$ when $x_\ell = y_\ell$ for all $\ell \in [\lambda] \setminus \{k\}$. Then by~\Cref{cor:TwoBidderWelfareGuarantee}, the total welfare loss across each coupled pair of instances is at least $m/30$, and hence the average welfare loss on $\calD_{i,u}$ is $m/60$. Since the total welfare in any instance is at most $u_{i,\vb{x}}([m]) + v_{-i,k,1}([m]) + u_{i,\vb{x}}(\{\dagger_i\}) + v_{-i,k,1}(\{\dagger_{-i}\}) \leq 22m$, losing $m/60$ welfare would mean $\MECH$ attains at most a $(1 - (m/60)/(22m)) < 0.9996$-approximation in expectation over $\calD_{i,u}$, a contradiction.
\end{proof}

Therefore, fix the choice of $i \in \{1, 2\}$ so that bidder $i$ is prescribed two different messages at the root. We will restrict our attention to the performance of $\MECH$ over $\calD_{i,u}$. Let $\sigma_i, \sigma_{-i}$ be the strategy prescriptions of $\MECH$, and let $\xi_i, \xi_{-i}$ be the first round message prescriptions of $\MECH$.

Since $\MECH$ and $\MECH'$ realize the same allocations and payments over the restricted domain $\calF_1 \times \calF_2$, the conclusions of~\Cref{lemma:SafePaymentLow,lemma:SafePaymentHigh} hold for $\MECH$ when $\MECH'$ is safe for $(i, k)$.

We now prove Step (2) of the proof sketch.

\begin{proposition} \label{prop:BobSignalsHisSpecialValue}
    For all $k \in [\lambda]$ such that $\MECH'$ is safe for $(i, k)$, $\xi_{-i}(v_{-i,k,0}) \ne \xi_{-i}(v_{-i,k,1})$.
\end{proposition}
\begin{proof}
Fix a choice of safe $k$, and suppose for contradiction that $\xi_{-i}(v_{-i,k,0}) = \xi_{-i}(v_{-i,k,1})$.

Let $\vb{x} \in \{0, 1\}^\lambda$, and let $f_i, g_i \in \calF_i$ be valuations such that $\xi_i(f_i) \ne \xi_i(g_i)$. WLOG, let $\xi_i(u_{i,\vb{x}}) \ne \xi_i(f_i)$. Then consider the contingent strategy $\psi_{-i}$ for bidder $-i$ which plays $\sigma_{-i}(v_{-i,k,1})$ when bidder $i$ sends $\xi_i(u_{i,\vb{x}})$ in the first round and plays $\sigma_{-i}(v_{-i,k,0})$ otherwise.

By~\Cref{lemma:SafePaymentHigh}, on the profile $\sigma_i(u_{i,\vb{x}}), \psi_{-i}$, either bidder $i$ does not receive $\dagger_i$ or their payment is at least $4m$. By~\Cref{lemma:SafePaymentLow}, on the profile $\sigma_i(f_i), \psi_{-i}$, bidder $i$ receives $\dagger_i$ and their payment is at most $2m$. Therefore, the utility of bidder $i$ holding valuation $u_{i,\vb{x}}$ playing their prescribed strategy yields at most $\max\{u_{i,\vb{x}}(M \setminus \{\dagger_i\}), u_{i,\vb{x}}(M) - 4m\} \leq 7m$ utility, while deviating to $\sigma_i(f_i)$ yields at least $u_{i,\vb{x}}(\{\dagger_i\}) - 2m = 8m$ utility. This contradicts $\MECH$ being DSIC.
\end{proof}

For all $\vb{x} \in \{0, 1\}^\lambda$ and $k \in [\lambda]$, let $S_i(\vb{x}, k), S_{-i}(\vb{x}, k)$ denote the allocation of $\MECH$ on the instance $u_{i,\vb{x}}, v_{-i,k,x_k}$. Let $\Delta_{i,\vb{x},k} = \Delta_{i,\vb{x},k}(S_i(\vb{x}, k), S_{-i}(\vb{x}, k))$. We now prove Step (3) of the proof sketch.

\begin{proposition} \label{prop:TwoBidderWelfareGuarantee}
    Suppose $\MECH'$ is safe for $(i, k)$. Then if there exist $\vb{x}, \vb{y} \in \{0, 1\}^\lambda$ such that $x_k = 0$, $y_k = 1$, and $\xi_i(u_{i,\vb{x}}) = \xi_i(u_{i,\vb{y}})$,
    \[
        \Delta_{i,\vb{x},k} + \Delta_{i,\vb{y},k} \quad \geq \quad m/30 \enspace .
    \]
\end{proposition}
\begin{proof}
By~\Cref{prop:BobSignalsHisSpecialValue}, $\xi_{-i}(v_{-i,k,0}) \ne \xi_{-i}(v_{-i,k,1})$. Therefore, consider the contingent strategy $\psi_i$ for bidder $i$ which plays $\sigma_i(u_{i,\vb{x}})$ if bidder $-i$ sends $\xi_{-i}(v_{-i,k,0})$ and $\sigma_i(u_{i,\vb{y}})$ otherwise. Let $S_0, p_0$ be the outcome for bidder $-i$ on the profile $\psi_i, \sigma_{-i}(v_{-i,k,0})$ and $S_1, p_1$ be the outcome for bidder $-i$ on the profile $\psi_i, \sigma_{-i}(v_{-i,k,1})$. By applying DSIC on each valuation, we have
\begin{align*}
    v_{-i,k,0}(S_0) - p_0 \quad &\geq \quad v_{-i,k,0}(S_1) - p_1 \\
    v_{-i,k,1}(S_1) - p_1 \quad &\geq \quad v_{-i,k,1}(S_0) - p_0 \enspace .
\end{align*}
Adding and canceling the payment terms,
\begin{align*}
    v_{-i,k,1}(S_1) - v_{-i,k,0}(S_1) \quad &\geq \quad v_{-i,k,1}(S_0) - v_{-i,k,0}(S_0) \\
    \abs{C_{k,0} \cap S_1}/3 + 5m \cdot \One(\dagger_i \in S_1) \quad &\geq \quad \abs{C_{k,0} \cap S_0}/3 + 5m \cdot \One(\dagger_i \in S_0) .
\end{align*}
There are two ways for the inequality to be satisfied: either $\dagger_i \in S_1$, or $\abs{C_{k,0} \cap S_0} \leq \abs{C_{k,0} \cap S_1}$. In the first case, the loss in welfare is already at least $5m$. In the second case, the loss in welfare is at least $m/30$ by~\Cref{cor:TwoBidderWelfareGuarantee}.
\end{proof}

\subsection{An Algorithm for INDEX}

Let $\vb{z} \in \{0, 1\}^\lambda$ be uniformly random. For all $\mu \in \{\xi_i(u_{i,\vb{x}}) : \vb{x} \in \{0, 1\}^\lambda\}$ and $k \in [\lambda]$, define
\[
    h(\mu, k) = \argmax_{b \in \{0, 1\}} \Pr[z_k = b \given \xi_i(u_{i,\vb{z}}) = \mu] \enspace ,
\]
which is the most likely value of $z_k$ if we observe message $\mu$ from bidder $i$ with valuation $u_{i,\vb{z}}$. We now show that $h(\mu, k)$ is a good estimator when $\MECH$ attains good welfare.

\begin{proposition} \label{prop:AliceLowersWelfareByHiding}
    If $\MECH'$ is safe for $(i, k)$, then
    \[
        \Pr[h(\xi_i(u_{i,\vb{z}}), k) \ne z_k] \quad \leq \quad 30/m \cdot \E[\Delta_{i,\vb{z},k}] \enspace .
    \]
\end{proposition}
\begin{proof}
We have
\begin{align*}
    \E[\Delta_{i,\vb{z},k}] \quad &= \quad \E\bigg[\E\bigg[\sum_{b=0}^1 \Big(\Pr[z_k = b \given \xi_i(u_{i,\vb{z}})] \cdot \E[\Delta_{i,\vb{z},k} \given \xi_i(u_{i,\vb{z}}),\, z_k = b]\Big) \bbgiven \xi_i(u_{i,\vb{z}})\bigg]\bigg] \\
    &\geq \quad \E\bigg[\E\bigg[\min_{b \in \{0,1\}} \Pr[z_k = b \given \xi_i(u_{i,\vb{z}})] \cdot \sum_{b=0}^1 \E[\Delta_{i,\vb{z},k} \given \xi_i(u_{i,\vb{z}}),\, z_k = b] \bbgiven \xi_i(u_{i,\vb{z}})\bigg]\bigg] \\
    &\geq \quad \E\bigg[\E\bigg[\frac{m}{30} \min_{b \in \{0,1\}} \Pr[z_k = b \given \xi_i(u_{i,\vb{z}})] \bbgiven \xi_i(u_{i,\vb{z}})\bigg]\bigg] \\
    &\geq \quad \E\bigg[\E\bigg[\frac{m}{30} \Pr[z_k \ne h(\xi_i(u_{i,\vb{z}}), k) \given \xi_i(u_{i,\vb{z}})] \bbgiven \xi_i(u_{i,\vb{z}})\bigg]\bigg] \\
    &= \quad \frac{m}{30} \Pr[z_k \ne h(\xi_i(u_{i,\vb{z}}), k)] \enspace ,
\end{align*}
where the third line comes from~\Cref{prop:TwoBidderWelfareGuarantee}.
\end{proof}

We can finally wrap up. Define the uniform INDEX problem to be INDEX on a uniformly random instance $(\vb{z}, \ell) \in \{0, 1\}^\lambda \times [\lambda]$. Solving uniform INDEX w.p. $1/2 + \Omega(1)$ requires $\Omega(\lambda) = 2^{\Theta(m)}$ communication~\cite{Ablayev96}. Consider the following algorithm for uniform INDEX:
\begin{enumerate}[topsep=1pt,itemsep=1pt]
    \item Ask Alice to send $\xi_i(u_{i,\vb{z}})$ to Bob.
    \item Bob outputs $h(\xi_i(u_{i,\vb{z}}), \ell)$.
\end{enumerate}

\begin{lemma} \label{lemma:Safety}
    Let $L = \{k \in [\lambda] : \text{$\MECH'$ is not safe for $(i, k)$}\}$. Then $\abs{L} < \lambda/10$.
\end{lemma}
\begin{proof}
Suppose otherwise for contradiction. Then on $\calD_{i,w}$, $\MECH'$ performs an incorrect allocation of a special item w.p.\ at least $1/20$. Mis-allocating either special item loses at least $m$ welfare, and the optimal welfare is at most $v_{-i,k,1}(M) \leq 16m$. Thus, $\MECH'$ attains at best a $(1 - (m/16)/20) < 0.9996$-approximation in expectation over $\calD_{i,w}$, a contradiction.
\end{proof}

Observe that $\E[\Delta_{i,\vb{z},\ell}]$ is the expected welfare loss of $\MECH$ on $\calD_{i,u}$. As the optimal welfare from $\calD_{i,u}$ is bounded by $22m$, we have $\E[\Delta_{i,\vb{z},\ell}] \leq 0.0004 \cdot 22m \leq m/100$. Therefore, by~\Cref{lemma:Safety} and~\Cref{prop:AliceLowersWelfareByHiding},
\[
    \Pr[h(\xi_i(u_{i,\vb{z}}), \ell) \ne z_\ell] \quad \leq \quad \Pr[\text{$\MECH'$ is not safe for $(i, \ell)$}] + 30/m \cdot \E[\Delta_{i,\vb{z},\ell}] \quad \leq \quad 2/5 \enspace .
\]

Thus, we have constructed a simultaneous algorithm which solves uniform INDEX w.p.\ $3/5$, so the simultaneous algorithm requires exponential communication. This implies that $\MECH$ (and therefore $\MECH''$) requires exponential communication, which proves~\Cref{thm:TwoBidderDSICHardness}.

\section{Simultaneous Algorithm for Two Matroid Rank Valuations}

In this section, we give a simultaneous algorithm for two matroid rank valuations which leverages submodularity to beat the tight simultaneous $3/4$-approximation for binary XOS valuations. Note that interactive algorithms beating a $3/4$-approximation for submodular valuations have been known for some time~\cite{FeigeV10}, but these algorithms need access to a fractional solution, which we do not know how to compute simultaneously. The main idea is to construct two algorithms: one which does well when the matroids are ``close to uniform,'' and one which does well when the matroids are far from uniform. Mixing these two algorithms allows us to do well on all instances.

\subsection{Matroid Preliminaries}

Below, we discuss properties of matroids we will need for our simultaneous algorithm.

\begin{definition}[Basis]
    Let $\calM \subseteq 2^M$ be a matroid. A set $S \in \calM$ is called an \emph{independent set}. A maximal independent set is called a \emph{basis}. Note that by the augmentation property, all bases of a matroid have size equal to the rank of the matroid.
\end{definition}

\begin{definition}[Polytopes]
    Let $\calM \subseteq 2^M$ be a matroid and $r$ be its rank function. Its \emph{independence polytope} is defined $\calP(\calM) = \{\vb{x} \in [0, 1]^m : \vb{x}(S) \leq r(S) \; \forall S \subseteq M\}$. In other words, $\calP(\calM)$ is the convex hull of the indicator vectors of the independent sets of $\calM$. The \emph{base polytope} is defined $\calP_B(\calM) = \{\vb{x} \in \calP(\calM) : \vb{x}(M) = r(M)\}$. Equivalently, $\calP_B(\calM)$ is the convex hull of the indicator vectors of the bases of $\calM$.
\end{definition}

\begin{fact}[\cite{Cunningham84}] \label{fact:PolytopeSampling}
    Let $\calM \subseteq 2^M$ be a matroid and $\vb{x} \in \calP_B(\calM)$. Then there exists an efficiently computable distribution $\calD$ over bases of $\calM$ where for all $j \in M$, $\Pr_{S \sim \calD}[j \in S] = x_j$.
\end{fact}

\begin{definition}[Restriction]
    Let $\calM \subseteq 2^M$ be a matroid. Its \emph{restriction} to $S$ is the set family $\calM \given S = \{T \subseteq S : T \in \calM\}$.
\end{definition}

\begin{definition}[Contraction]
    Let $\calM \subseteq 2^M$ be a matroid and $r$ be its rank function. Its \emph{contraction} by $S$ is the matroid $\calM \contract S$ given by the rank function $r_S(T) = r(S \cup T) - r(S)$.
\end{definition}

\begin{definition}[Dual]
    Given a matroid $\calM \subseteq 2^M$ with rank function $r : 2^M \to \mathbb{Z}_+$, its \emph{dual} $\calM^*$ is defined as the downward-closure of the set family $\{M \setminus S : S \in \calM, \abs{S} = r(M)\}$. In other words, the complements of the bases of $\calM$ are the maximal sets of $\calM^*$.
\end{definition}

\begin{fact} \label{fact:DualRank}
    Given a matroid $\calM \subseteq 2^M$ with rank function $r$ and dual rank function $r^*$, we have $r^*(S) = \abs{S} - r(M) + r(M \setminus S)$.
\end{fact}

\begin{fact}
    The restrictions, contractions, and dual of a matroid are matroids.
\end{fact}

We now define the principal partition of a matroid. While principal partitions of matroids and their extensions have been studied extensively over the years~\cite{NarayananV81,Fujishige09}, our definitions follow the formulation of~\cite{Soto13}, which is tailored to a theoretical computer science audience.

\begin{definition}[Principal Partition]
    Let $\calM \subseteq 2^M$ be a matroid with rank function $r$, and let $U_0 \subseteq M$ be the maximal set satisfying $r(U_0) = 0$. For sets $U_1, \dots, U_s$, define $U_{< k} = \bigcup_{\ell < k} U_\ell$ and $U_{\leq k} = \bigcup_{\ell \leq k} U_\ell$ for $k \in [s]$. Then there exists a partition $U_0, U_1, \dots, U_s$ and values $0 = \mu_0 < \mu_1 < \dots < \mu_s \leq 1$ satisfying the following:
    \begin{itemize}[topsep=1pt,itemsep=1pt]
        \item For $\nu > 0$, define $f(S, \nu) = r(S)/\nu - \abs{S}$. Then $f(\cdot, \nu) : 2^M \to \mathbb{R}$ admits more than one minimizer if and only if $\nu \in \{\mu_1, \dots, \mu_s\}$.
        \item For every $k \in [s]$, the unique minimal and unique maximal minimizers of $f(\cdot, \mu_k)$ are $U_{< k}$ and $U_{\leq k}$, respectively.
    \end{itemize}
    The partition $U_0, U_1, \dots, U_s$ is called the \emph{principal partition} of $\calM$ with densities $\mu_0, \mu_1, \dots, \mu_s$. The \emph{principal vector} is $\vb{p} \in [0, 1]^m$, defined $p_j = \mu_k$ for all $k \in [s] \cup \{0\}$ and $j \in U_k$.
\end{definition}

\begin{fact}[Uniform Density] \label{fact:PrincipalPartition}
    Let $\calM \subseteq 2^M$ be a matroid with rank function $r$. Let its principal partition be given by $U_0, U_1, \dots, U_s$ and $\mu_0, \mu_1, \dots, \mu_s$, and let $\vb{p}$ be its principal vector. Let $r_k$ be the rank function of $(\calM \contract U_{< k}) \given U_k$. Then for all $k \in [s]$ and $S \subseteq U_k$, $r_k(S) \geq \mu_k\abs{S}$, with equality holding at $S = U_k$. As a direct consequence, $r(U_{\leq k}) = \vb{p}(U_{\leq k})$ for all $k \in [s]$.
\end{fact}

\begin{definition}[Self-Union]
    Let $\calM \subseteq 2^M$ be a matroid. Then its \emph{self-union} is defined $\calM \lor \calM = \{S \cup T \subseteq M : S, T \in \calM\}$.
\end{definition}

\begin{fact}[Matroid Union Theorem~\cite{Edmonds68}]
    Let $\calM \subseteq 2^M$ be a matroid with rank function $v$, and let $v \lor v$ be the rank function of $\calM \lor \calM$. Then
    \[
        (v \lor v)(S) \quad = \quad \min_{T \subseteq S} \Big(\abs{S \setminus T} + 2v(T)\Big) \enspace .
    \]
\end{fact}

\begin{fact}[\cite{FeigeV10}] \label{fact:HalfIntegralRounding}
    Let $u, v : 2^M \to \mathbb{R}_+$ be submodular. Let $A, B, C, D \subseteq M$ be such that no item is contained in more than two of the four sets. Then there is a randomized allocation depending only on $A, B, C, D$ which obtains welfare at least $(5/12)(u(A) + u(B) + v(C) + v(D))$ in expectation.
\end{fact}

\begin{corollary} \label{cor:HalfIntegralRounding}
    Let $\calU, \calV \subseteq 2^M$ be matroids and let $u, v : 2^M \to \mathbb{Z}_+$ be their rank functions. Let $A, B \in \calU$ and $C, D \in \calV$. Then there is a randomized allocation depending only on $A, B, C, D$ which obtains welfare at least $(\abs{A \cup B} + \abs{C \cup D})/2 - \abs{(A \diff B) \cap (C \diff D)}/6$ in expectation.
\end{corollary}
\begin{proof}
Define the disjoint sets
\begin{align*}
    X \quad &= \quad (A \cap B) \cup ((A \cup B) \setminus (C \cup D)) \\
    Y \quad &= \quad ((C \cap D) \setminus (A \cap B)) \cup ((C \cup D) \setminus (A \cup B)) \\
    Z \quad &= \quad (A \diff B) \cap (C \diff D) \enspace .
\end{align*}
Our rounding scheme allocates $X$ to $u$ and $Y$ to $v$, then runs the rounding scheme of~\Cref{fact:HalfIntegralRounding} on $A \cap Z$, $B \cap Z$, $C \cap Z$, $D \cap Z$.

Let $u_X, v_Y$ be the rank functions of $\calU \contract X, \calV \contract Y$. Our rounding scheme obtains welfare
\begin{align*}
    w \quad &\geq \quad u(X) + v(Y) + \frac{5}{12}\Big(u_X(A \cap Z) + u_X(B \cap Z) + v_Y(C \cap Z) + v_Y(D \cap Z)\Big) \\
    &\geq \quad \frac{1}{2}\Big(u(A \cap (X \cup Z)) + u(B \cap (X \cup Z)) + v(C \cap (Y \cup Z)) + v(D \cap (Y \cup Z))\Big) \\
    & \quad \qquad\qquad - \frac{1}{12}\Big(u_X(A \cap Z) + u_X(B \cap Z) + v_Y(C \cap Z) + v_Y(D \cap Z)\Big) \\
    &\geq \quad \frac{1}{2}\Big(u(A \cap (X \cup Z)) + u(B \cap (X \cup Z)) + v(C \cap (Y \cup Z)) + v(D \cap (Y \cup Z))\Big) \\
    & \quad \qquad\qquad - \frac{1}{12}\Big(\abs{A \cap Z} + \abs{B \cap Z} + \abs{C \cap Z} + \abs{D \cap Z}\Big) \\
    &\geq \quad \frac{1}{2}\Big(u(A \cap (X \cup Z)) + u(B \cap (X \cup Z)) + v(C \cap (Y \cup Z)) + v(D \cap (Y \cup Z))\Big) - \frac{1}{6} \abs{Z} \enspace .
\end{align*}

Observe that $X \cup Z$ contains every element of $A \cap B$ and every element of $(A \setminus B) \cap (C \diff D)$, and so from $A$, it is only missing the elements of $(A \setminus B) \cap (C \cap D)$. The same reasoning holds for $B$, so
\begin{align} \label{eq:SimultaneousWelfareU}
\begin{split}
    & u(A \cap (X \cup Z)) + u(B \cap (X \cup Z)) \\
    & \qquad\qquad\qquad \geq \quad \abs{A \setminus ((A \setminus B) \cap (C \cap D))} + \abs{B \setminus ((B \setminus A) \cap (C \cap D))} \\
    & \qquad\qquad\qquad = \quad \abs{A} + \abs{B} - \abs{(A \diff B) \cap (C \cap D)} \\
    & \qquad\qquad\qquad = \quad \abs{A} + \abs{B} - \abs{(A \cup B) \cap (C \cap D)} + \abs{A \cap B \cap C \cap D} \\
    & \qquad\qquad\qquad \geq \quad \abs{A} + \abs{B} - \abs{C \cap D} + \abs{A \cap B \cap C \cap D} \enspace .
\end{split}
\end{align}
The same argument on $Y \cup Z \cup (A \cap B \cap C \cap D)$ and $C, D$ yields
\begin{align} \label{eq:SimultaneousWelfareV}
\begin{split}
    & v(C \cap (Y \cup Z)) + v(D \cap (Y \cup Z)) \\
    & \qquad\qquad\qquad \geq \quad \abs{C} + \abs{D} - \abs{(A \cap B) \cap (C \diff D)} - 2\abs{A \cap B \cap C \cap D} \\
    & \qquad\qquad\qquad \geq \quad \abs{C} + \abs{D} - \abs{A \cap B} - \abs{A \cap B \cap C \cap D} \enspace .
\end{split}
\end{align}

Substituting~\eqref{eq:SimultaneousWelfareU} and~\eqref{eq:SimultaneousWelfareV} into the welfare bound yields
\begin{align*}
    w \quad &\geq \quad \frac{1}{2}\Big(\abs{A} + \abs{B} - \abs{C \cap D} + \abs{A \cap B \cap C \cap D} \\
    & \quad \qquad\qquad + \abs{C} + \abs{D} - \abs{A \cap B} - \abs{A \cap B \cap C \cap D}\Big) - \frac{1}{6}\abs{Z} \\
    &= \quad \frac{1}{2}\Big(\abs{A \cup B} + \abs{C \cup D}\Big) - \frac{1}{6}\abs{Z} \enspace . \qedhere
\end{align*}
\end{proof}

\subsection{Algorithm}

Let $u, v : 2^M \to \mathbb{Z}_+$ be the rank functions of the matroids $\calU, \calV \subseteq 2^M$. Let $U_1, \dots, U_s$ be the principal partition of $\calU$ with densities $\mu_1, \dots, \mu_s$, and let $V_1, \dots, V_t$ be the principal partition of $\calV$ with densities $\nu_1, \dots, \nu_t$. Let $\vb{p}, \vb{q} \in [0, 1]^m$ be the principal vectors of $\calU, \calV$, respectively. We begin by showing that for every matroid, it is possible to sample ``maximally disjoint'' bases with marginals equal to the principal vector.

\begin{definition}
    A distribution over pairs of bases $A, B \in \calU$ is \emph{balanced} if for all $j \in M$,
    \[
        \E[\mathbbm{1}(j \in A) + \mathbbm{1}(j \in B)] = 2p_j \qquad \text{and} \qquad \mathbbm{1}(j \in A) + \mathbbm{1}(j \in B) \in \{\lfloor 2p_j \rfloor, \lceil 2p_j \rceil\} \enspace .
    \]
\end{definition}

\begin{proposition} \label{prop:PairSampling}
    It is possible to sample from a balanced distribution for any matroid.
\end{proposition}
\begin{proof}
For $k = 1, \dots, s$, we will sample a pair of bases $A_k, B_k \in (\calU \contract U_{< k}) \given U_k$, and output the pair of bases $A_1 \cup \dots \cup A_s, B_1 \cup \dots \cup B_s$.~\Cref{fact:PrincipalPartition} yields inductively that $A_1 \cup \dots \cup A_k$ is a basis of $\calU \given U_{\leq k}$ for every $k \in [s]$, and in particular that the final outputs are bases of $\calU$. Therefore, it suffices to describe the sampling procedure of bases $A, B \in \calU \given U_1$ satisfying the conditions for all $j \in U_1$. Recall that $p_j = \mu_1$ for all $j \in U_1$.

Suppose $\mu_1 \leq 1/2$. Let $u \lor u$ be the rank function of $\calU \lor \calU$. Then by the Matroid Union Theorem and~\Cref{fact:PrincipalPartition}, we have for all $S \subseteq U_1$ that
\begin{equation} \label{eq:RankRelation}
    (u \lor u)(S) \quad = \quad \min_{T \subseteq S} \Big(\abs{S \setminus T} + 2u(T)\Big) \quad \geq \quad \min_{T \subseteq S} \Big(2\mu_1\abs{S \setminus T} + 2\mu_1\abs{T}\Big) \quad = \quad 2\mu_1\abs{S} \enspace .
\end{equation}

By definition of self-union, we have $(u \lor u)(U_1) \leq 2u(U_1) = 2\mu_1\abs{U_1}$. On the other hand,~\eqref{eq:RankRelation} yields that $(u \lor u)(U_1) \geq 2\mu_1\abs{U_1}$, so $(u \lor u)(U_1) = 2\mu_1\abs{U_1}$. This, combined with the rank lower bound on all $S \subseteq U_1$ given by~\eqref{eq:RankRelation}, shows that the vector $2\mu_1 \mathbf{1}$ is contained in $\calP_B((\calU \lor \calU) \given U_1)$.

Hence, by~\Cref{fact:PolytopeSampling}, there exists a distribution $\calD$ over bases of $(\calU \lor \calU) \given U_1$ where $\Pr_{S \sim \calD}[j \in S] = 2\mu_1$ for all $j \in U_1$. By definition of $\calU \lor \calU$ and the fact that $(u \lor u)(U_1) = 2u(U_1)$, each basis supported in $\calD$ can be partitioned into a pair of bases of $\calU \given U_1$. As this procedure for sampling pairs of bases of $\calU \given U_1$ includes each item $j \in U_1$ in exactly one basis w.p.\ $2\mu_1$ and in neither basis otherwise, it satisfies the desired condition.

If $\mu_1 > 1/2$, then consider the dual matroid $(\calU \given U_1)^* \subseteq 2^{U_1}$ with rank function $u^* : 2^{U_1} \to \mathbb{Z}_+$. By~\Cref{fact:DualRank,fact:PrincipalPartition}, we have for all $S \subseteq U_1$ that
\[
    u^*(S) \quad = \quad \abs{S} - u(U_1) + u(U_1 \setminus S) \quad \geq \quad \abs{S} - \mu_1\abs{U_1} + \mu_1\abs{U_1 \setminus S} \quad = \quad (1 - \mu_1)\abs{S} \enspace .
\]

Since $1 - \mu_1 < 1/2$, we may repeat the argument above on the dual matroid, and take complements to obtain a pair of bases in the original matroid. In this case, an item will be excluded from the bases in the dual matroid (and hence included in both bases of the original matroid) w.p.\ $1 - 2(1 - \mu_1) = 2\mu_1 - 1$, and will be included in exactly one basis of the original matroid otherwise. This yields the correct support and the correct expectation $2(2\mu_1 - 1) + 1(1 - (2\mu_1 - 1)) = 2\mu_1$.
\end{proof}

We now describe the simultaneous algorithm (which only uses $4m$ communication).
\begin{enumerate}[topsep=1pt,itemsep=1pt]
    \item Bidder $u$ samples $A, B \in \calU$ from a balanced distribution and sends it to the referee.
    \item Bidder $v$ (independently from bidder $u$) samples $C, D \in \calV$ from a balanced distribution and sends it to the referee.
    \item The referee mixes between the following procedures:
    \begin{enumerate}[topsep=1pt,itemsep=1pt]
        \item Choose a uniformly random set among $A, B, C, D$ and allocate it to the sender and the complement to the other bidder.
        \item Run the rounding scheme of~\Cref{cor:HalfIntegralRounding} on the sets $A, B, C, D$.
    \end{enumerate}
\end{enumerate}

\subsection{Analysis}

We now analyze the approximation guarantee of each procedure. Let $X \in \calU$ and $Y \in \calV$ be disjoint sets attaining the optimal welfare, let $W = X \cup Y$, and define
\[
    \Delta_{\vb{p}} \;\; = \sum_{j \in W} (1/2 - p_j)^2 \qquad \text{and} \qquad \Delta_{\vb{q}} \;\; = \sum_{j \in W} (1/2 - q_j)^2 \enspace .
\]
Assume WLOG that $W \ne \emptyset$, and define $\Delta = \Delta_{\vb{p}} + \Delta_{\vb{q}}$ and $\delta = \sqrt{2\Delta/\abs{W}}$. We will show that the contention resolution scheme in Procedure (a) works well when $\Delta$ is large, while the submodular rounding scheme in Procedure (b) works well when $\Delta$ is small.

\begin{lemma} \label{lemma:Guarantee1}
    In expectation, Procedure (a) guarantees welfare at least
    \[
        \bigg(\frac{3}{4} + \frac{\delta^2}{4}\bigg) \abs{W} \enspace .
    \]
\end{lemma}
\begin{proof}
Suppose $R \in \{A, B\}$ is the set chosen by Procedure (a). Observe that $R$ is drawn from a distribution over bases with marginals $\vb{p}$. Recall that $U_0, U_1, \dots, U_s$ is the principal partition of $\calU$ with densities $0 = \mu_0 < \mu_1 < \dots < \mu_s \leq 1 \eqqcolon \mu_{s+1}$. By~\Cref{fact:PrincipalPartition},
\begin{align} \label{eq:PrincipalPartition}
\begin{split}
    \sum_{j \in X} (1 - p_j) \quad &= \quad \sum_{k \in [s]} (1 - \mu_k)\abs{X \cap U_k} \\
    &= \quad \sum_{k \in [s]} (\mu_{k+1} - \mu_k)\abs{X \cap U_{\leq k}} \\
    &\leq \quad \sum_{k \in [s]} (\mu_{k+1} - \mu_k) u(U_{\leq k}) \\
    &= \quad \sum_{k \in [s]} (\mu_{k+1} - \mu_k) \vb{p}(U_{\leq k}) \quad = \quad \sum_{j \in M} p_j(1 - p_j) \enspace .
\end{split}
\end{align}

Then we can lower bound the expected welfare of the algorithm by
\begin{align*}
    \abs{Y} + \E[\abs{R \cap \closure{Y}}] \quad &= \quad \abs{W} - \abs{X} + \E[\abs{R \cap X}] + \E[\abs{R \cap \closure{W}}] \\
    &= \quad \abs{W} - \sum_{j \in X} (1 - p_j) + \sum_{j \not\in W} p_j \\
    &\geq \quad \abs{W} - \sum_{j \in M} p_j (1 - p_j) + \sum_{j \not\in W} p_j \\
    &\geq \quad \abs{W} - \sum_{j \in W} p_j (1 - p_j) \\
    &= \quad \frac{3}{4}\abs{W} + \sum_{j \in W} (p_j - 1/2)^2 \\
    &= \quad \frac{3}{4}\abs{W} + \Delta_{\vb{p}}\enspace ,
\end{align*}
where the third line is by~\eqref{eq:PrincipalPartition}. Repeating the argument for $R \in \{C, D\}$ and averaging completes the proof.
\end{proof}

\begin{lemma} \label{lemma:Guarantee2}
    In expectation, Procedure (b) guarantees welfare at least
    \[
        \bigg(\frac{5}{6} - \frac{\delta}{6} - \frac{\delta^2}{6}\bigg) \abs{W} \enspace .
    \]
\end{lemma}
\begin{proof}
Define $\pi_j = \Pr[j \in A \cup B] + \Pr[j \in C \cup D] - \Pr[j \in (A \diff B) \cap (C \diff D)]/3$. By~\Cref{cor:HalfIntegralRounding}, Procedure (b) guarantees expected welfare at least
\[
    \frac{1}{2} \E[\abs{A \cup B}] + \frac{1}{2} \E[\abs{C \cup D}] - \frac{1}{6} \E[\abs{(A \diff B) \cap (C \diff D)}] \quad = \quad \frac{1}{2} \vb{\pi}(M) \enspace .
\]
Since we sample $A, B \in \calU$ from a balanced distribution, we have
\begin{align*}
    \Pr[j \in A \cup B] \quad &= \quad \min\{1, 2p_j\} \quad = \quad p_j + 1/2 - \abs{1/2 - p_j} \\
    \Pr[j \in A \cap B] \quad &= \quad \max\{0, 2p_j - 1\} \quad = \quad p_j - 1/2 + \abs{1/2 - p_j} \\
    \Pr[j \in A \diff B] \quad &= \quad \Pr[j \in A \cup B] - \Pr[j \in A \cap B] \quad = \quad 1 - 2\abs{1/2 - p_j} \enspace ,
\end{align*}
and similarly for $C, D$. Then observe that we have the lower bound
\begin{align}
\begin{split} \label{eq:ItemContributionLowerBound}
   \pi_j \quad &= \quad p_j + q_j + 1 - \abs{1/2 - p_j} - \abs{1/2 - q_j} - \frac{1}{3}(1 - 2\abs{1/2 - p_j})(1 - 2\abs{1/2 - q_j}) \\
   &= \quad p_j + q_j + \frac{2}{3} - \frac{1}{3}\abs{1/2 - p_j} - \frac{1}{3}\abs{1/2 - q_j} - \frac{4}{3}\abs{1/2 - p_j}\abs{1/2 - q_j} \\
   &\geq \quad p_j + q_j + \frac{2}{3} - \frac{1}{6} - \frac{1}{6} - \frac{1}{3} \\
   &= \quad p_j + q_j \enspace .
\end{split}
\end{align}
Finally, since $X \in \calU$ and $Y \in \calV$,
\begin{equation} \label{eq:TotalContributionBound}
    \abs{W} \quad \leq \quad \vb{p}(M) + \vb{q}(M) \quad = \quad \vb{p}(W) + \vb{q}(W) + \vb{p}(\closure{W}) + \vb{q}(\closure{W}) \enspace .
\end{equation}
Therefore, combining~\Cref{eq:ItemContributionLowerBound,eq:TotalContributionBound} yields
\begin{align*}
    \vb{\pi}(M) \quad &= \quad \vb{\pi}(W) + \vb{\pi}(\closure{W}) \quad \geq \quad \vb{\pi}(W) + \vb{p}(\closure{W}) + \vb{q}(\closure{W}) \quad \geq \quad \vb{\pi}(W) + \abs{W} - \vb{p}(W) - \vb{q}(W) \enspace .
\end{align*}

Substituting and simplifying, we have
\begin{align*}
    \vb{\pi}(M) \quad &\geq \quad \sum_{j \in W} \Big(\pi_j + 1 - p_j - q_j\Big) \\
    &= \quad \sum_{j \in W} \bigg(\frac{5}{3} - \frac{1}{3} \abs{1/2-p_j} - \frac{1}{3} \abs{1/2-q_j} - \frac{4}{3} \abs{1/2-p_j}\abs{1/2-q_j}\bigg) \\
    &\geq \quad \sum_{j \in W} \bigg(\frac{5}{3} - \frac{1}{3} \Big(\abs{1/2-p_j} + \abs{1/2-q_j}\Big) - \frac{2}{3} \Big(\abs{1/2-p_j}^2 + \abs{1/2-q_j}^2\Big)\bigg) \\
    &\geq \quad \frac{5}{3} \abs{W} - \frac{1}{3} \sum_{j \in W} \Big(\abs{1/2-p_j} + \abs{1/2-q_j}\Big) - \frac{2}{3} \Delta \\
    &\geq \quad \frac{5}{3} \abs{W} - \frac{1}{3} \sqrt{2\Delta\abs{W}} - \frac{2}{3} \Delta \enspace ,
\end{align*}
where the third line follows by $a^2 + b^2 \geq 2ab$, and the fifth line follows by Cauchy-Schwarz. Dividing by two and substituting in $\delta$ completes the proof.
\end{proof}

The optimized algorithm chooses Procedure (a) w.p.\ $(4+\sqrt{6})/10$ and Procedure (b) otherwise. By~\Cref{lemma:Guarantee1,lemma:Guarantee2}, this yields an approximation ratio of at least
\begin{align*}
    & \frac{4 + \sqrt{6}}{10}\bigg(\frac{3}{4} + \frac{\delta^2}{4}) + \frac{6 - \sqrt{6}}{10}\bigg(\frac{5}{6} - \frac{\delta}{6} - \frac{\delta^2}{6}\bigg) \\
    & \qquad\qquad\qquad\qquad\qquad = \quad \frac{96 - \sqrt{6}}{120} - \frac{6 - \sqrt{6}}{60} \delta + \frac{\sqrt{6}}{24} \delta^2 \\
    & \qquad\qquad\qquad\qquad\qquad = \quad \frac{96 - \sqrt{6}}{120} + \frac{\sqrt{6}}{24}\bigg(\delta^2 - 2\frac{\sqrt{6} - 1}{5}\delta\bigg) \\
    & \qquad\qquad\qquad\qquad\qquad = \quad \frac{96 - \sqrt{6}}{120} - \frac{\sqrt{6}}{24}\bigg(\frac{\sqrt{6} - 1}{5}\bigg)^2 + \frac{\sqrt{6}}{24}\bigg(\delta - \frac{\sqrt{6} - 1}{5}\bigg)^2 \\
    & \qquad\qquad\qquad\qquad\qquad \geq \quad \frac{96 - \sqrt{6}}{120} - \frac{\sqrt{6}}{24}\bigg(\frac{\sqrt{6} - 1}{5}\bigg)^2 \\
    & \qquad\qquad\qquad\qquad\qquad = \quad \frac{5(96 - \sqrt{6}) - 6\sqrt{6} + 12 - \sqrt{6}}{600} \\
    & \qquad\qquad\qquad\qquad\qquad = \quad \frac{492 - 12\sqrt{6}}{600} \\
    & \qquad\qquad\qquad\qquad\qquad = \quad \frac{41 - \sqrt{6}}{50} \enspace ,
\end{align*}
which completes the proof of~\Cref{thm:TwoBidderSimultaneousAlgorithm}.

\section{DSIC Mechanism for Two Submodular Valuations}

We first define some properties of submodular valuations. When convenient, we may denote a singleton set $\{j\}$ simply by $j$.

\begin{definition}
    For a submodular valuation $v : 2^M \to \mathbb{R}_+$, the \emph{contraction} of $v$ on $T \subseteq M$ is denoted $v(\cdot \given T) : 2^M \to \mathbb{R}_+$, where $v(S \given T) = v(S \cup T) - v(T)$.
\end{definition}

\begin{fact}
    Submodularity is preserved under contraction.
\end{fact}

\begin{definition}[Extended-Parallel Second-Price Auction]
    The \emph{extended-parallel second-price auction} (EP2A) is the deterministic mechanism which does the following:
    \begin{enumerate}[topsep=1pt,itemsep=1pt]
        \item Ask each bidder $i$ to report $v_i(j)$ for all $j \in M$.
        \item For all $j \in M$, let $p_j$ be the second highest reported value for $j$.
        \item Let $M_i$ be the set of items for which bidder $i$ is the unique highest bidder. Allocate to each bidder $i$ the set $S_i \subseteq M_i$ which maximizes $v_i(S_i) - \vb{p}(S_i)$.
    \end{enumerate}
\end{definition}

\begin{fact}[\cite{FikiorisF20}] \label{fact:EP2ADSIC}
    The EP2A is DSIC for subadditive valuations.
\end{fact}

We now describe a DSIC mechanism for two submodular valuations. Define $\varphi = (\sqrt{5}-1)/2$ to be the inverse golden ratio. Let $v_1, v_2 : 2^M \to \mathbb{R}_+$ be submodular. For disjoint sets $S, T \subseteq M$, define $v_i(S \given T) = v_i(S \cup T) - v_i(T)$ to be the marginal value of $S$ when holding $T$.
\begin{enumerate}[topsep=1pt,itemsep=1pt]
    \item Sample a random partition $A_1, A_2, B$, where independently for each $j \in M$, assign $j$ to $A_1, A_2, B$ w.p.\ $\varphi^2, \varphi^2, \varphi^3$, respectively. Note that $2\varphi^2 + \varphi^3 = 1$.
    \item Ask each bidder $i$ to report $v_i(j \given A_i)$ for all $j \in B$.
    \item Let $p_j = \min_{i \in \{1, 2\}} \{v_i(j \given A_i)\}$ for all $j \in B$, and define $B_i = \{j \in B : v_i(j \given A_i) > p_j\}$.
    \item Allocate $A_i \cup S_i$ to bidder $i$, where $S_i = \argmax_{S \subseteq B_i} \{v_i(S \given A_i) - \vb{p}(S)\}$. Bidder $i$ pays $\vb{p}(S_i)$.
\end{enumerate}
Since the only communication is $2\abs{B} \leq 2m$ marginal values and two demand queries, the mechanism runs in polynomial communication. The mechanism is IR because each bidder always has the option to purchase nothing from $B$, and the mechanism is NNT because $\vb{p}$ has non-negative entries. Finally, fixing any partition $A_1, A_2, B$, the mechanism runs the EP2A on items $B$ and valuations $v_1(\cdot \given A_1), v_2(\cdot \given A_2)$. Since submodularity is preserved under contraction,~\Cref{fact:EP2ADSIC} tells us that the mechanism is DSIC for all fixed partitions. As the only randomness is in the choice of partition, the mechanism is universally DSIC.

\subsection{Analysis}

Let $X_1, X_2$ be a partition of $M$ that attains the optimal welfare. Define $\vb{w}_i \in \mathbb{R}_+^m$ so that $\vb{w}_i(S) \leq v_i(S \cap X_i)$ for all $S \subseteq M$, with equality holding at $S = M$. Also define $\vb{w}_{i,A_i} \in \mathbb{R}_+^m$ so that $\vb{w}_{i,A_i}(S) \leq v_i(S \cap X_i \given A_i)$ for all $S \subseteq M \setminus A_i$, with equality holding at $S = M \setminus A_i$. Such vectors exist because all submodular valuations are XOS.

We define $\Delta_i = \E[v_i(A_i) - \vb{w}_i(A_i)]$ and $\Delta = \Delta_1 + \Delta_2$. Roughly speaking, $\Delta_i$ captures how much value the random endowment $A_i$ gives to bidder $i$ over the worst-case additive function with value only on $X_i$. We will show that when $\Delta$ is large, the contribution from $A_1, A_2$ is large, and when $\Delta$ is small, the contribution from $B$ is large.

Fix a realization of $A_1, A_2, B$. Consider the following deviation: when asked for $v_i(j \given A_i)$, report $w_{i,A_i,j}$ instead. Then since $w_{i,A_i,j} = 0$ for $j \not\in X_i$, bidder $i$ receives utility at least
\[
    v_i(A_i) + \sum_{j \in B \cap X_i} \Big(w_{i,A_i,j} - v_{-i}(j \given A_{-i})\Big) \enspace .
\]
Therefore, by DSIC, the utility of truthful play must be at least as much. By NNT, utility lower bounds value. Therefore, summing and taking an expectation yields the welfare lower bound
\begin{equation} \label{eq:SubmodularWelfareBound}
    \sum_{i=1}^2 \E[v_i(A_i)] + \sum_{i=1}^2 \E[\vb{w}_{i,A_i}(B \cap X_i)] - \sum_{i=1}^2 \E\bigg[\sum_{j \in B \cap X_i} v_{-i}(j \given A_{-i})\bigg] \enspace .
\end{equation}

\begin{lemma} \label{lemma:EndowmentWelfare}
    \[
        \E[v_i(A_i)] \quad \geq \quad \varphi^2 \cdot v_i(X_i) + \Delta_i \enspace .
    \]
\end{lemma}
\begin{proof}
By definition of $\Delta_i$ and the fact that each item is in $A_i$ w.p.\ $\varphi^2$,
\[
    \E[v_i(A_i)] \quad = \quad \E[\vb{w}_i(A_i)] + \Delta_i \quad = \quad \varphi^2 \cdot \vb{w}_i(X_i) + \Delta_i \quad = \quad \varphi^2 \cdot v_i(X_i) + \Delta_i \enspace . \qedhere
\]
\end{proof}

\begin{lemma} \label{lemma:EP2AWelfare}
    \[
        \E[\vb{w}_{i,A_i}(B \cap X_i)] \quad \geq \quad \frac{\varphi}{1 + \varphi} \E[v_i(X_i \given A_i)] \enspace .
    \]
\end{lemma}
\begin{proof}
Conditioned on $A_i$, each remaining item is in $B$ w.p.\ $\varphi^3/(\varphi^2+\varphi^3) = \varphi/(1+\varphi)$, so
\begin{align*}
    \E[\vb{w}_{i,A_i}(B \cap X_i)] \quad &= \quad \E\Big[\E[\vb{w}_{i,A_i}(B \cap X_i) \given A_i]\Big] \\
    &= \quad \frac{\varphi}{1+\varphi} \E\Big[\E[\vb{w}_{i,A_i}(X_i \setminus A_i) \given A_i]\Big] \\
    &= \quad \frac{\varphi}{1+\varphi} \E\Big[\E[v_i(X_i \given A_i) \given A_i]\Big] \\
    &= \quad \frac{\varphi}{1+\varphi} \E[v_i(X_i \given A_i)] \enspace . \qedhere
\end{align*}
\end{proof}

\begin{lemma} \label{lemma:EP2APayment}
    \[
        \E\bigg[\sum_{j \in B \cap X_i} v_{-i}(j \given A_{-i})\bigg] \quad \leq \quad \varphi \cdot \Delta_{-i} \enspace .
    \]
\end{lemma}
\begin{proof}
By the diminishing marginals property of submodular valuations,
\begin{equation} \label{eq:SubmodularContributions}
    v_{-i}(A_{-i} \cap X_{-i}) + \sum_{j \in A_{-i} \cap X_i} v_{-i}(j \given A_{-i} \setminus \{j\}) \quad \leq \quad v_{-i}(A_{-i}) \enspace .
\end{equation}
Then we have
\begin{align*}
    \E\bigg[\sum_{j \in B \cap X_i} v_{-i}(j \given A_{-i})\bigg] \quad &= \quad \varphi^3 \sum_{j \in X_i} \E[v_{-i}(j \given A_{-i}) \given j \in B] \\
    &= \quad \varphi^3 \sum_{j \in X_i} \E[v_{-i}(j \given A_{-i} \setminus \{j\}) \given j \in A_{-i}] \\
    &= \quad \varphi \cdot \E\bigg[\sum_{j \in A_{-i} \cap X_i} v_{-i}(j \given A_{-i} \setminus \{j\})\bigg] \\
    &\leq \quad \varphi \cdot \E[v_{-i}(A_{-i}) - v_{-i}(A_{-i} \cap X_{-i})] \\
    &\leq \quad \varphi \cdot \Delta_{-i} \enspace ,
\end{align*}
where the first line is because each item is in $B$ w.p.\ $\varphi^3$, the second line is because items are independent, the third line is because each item is in $A_{-i}$ w.p.\ $\varphi^2$, and the fourth line is by~\eqref{eq:SubmodularContributions}.
\end{proof}

Substituting \Cref{lemma:EP2AWelfare,lemma:EP2APayment} into~\eqref{eq:SubmodularWelfareBound} yields a welfare lower bound of
\begin{align*}
    & \E[v_1(A_1) + v_2(A_2)] + \frac{\varphi}{1 + \varphi} \E[v_1(X_1 \given A_1) + v_2(X_2 \given A_2)] - \varphi \cdot \Delta \\
    & \qquad\qquad \quad \geq \quad \frac{1}{1 + \varphi} \E[v_1(A_1) + v_2(A_2)] + \frac{\varphi}{1 + \varphi} \E[v_1(X_1 \cup A_1) + v_2(X_2 \cup A_2)] - \varphi \cdot \Delta \\
    & \qquad\qquad \quad \geq \quad \frac{\varphi^2}{1 + \varphi} (v_1(X_1) + v_2(X_2)) + \frac{1}{1 + \varphi} \Delta + \frac{\varphi}{1 + \varphi} (v_1(X_1) + v_2(X_2)) - \varphi \cdot \Delta \\
    & \qquad\qquad \quad = \quad \varphi (v_1(X_1) + v_2(X_2)) + \frac{1 - \varphi - \varphi^2}{1 + \varphi} \Delta \enspace ,
\end{align*}
where the third line comes from~\Cref{lemma:EndowmentWelfare}. Observing that $1 - \varphi - \varphi^2 = 0$ completes the proof of~\Cref{thm:TwoBidderDSICMechanism}.

\end{document}